\documentclass{sig-alternate-10}

\usepackage{tikz}
\usetikzlibrary{arrows,automata,snakes,positioning,shapes,shadows,calc}
\usepackage{amsmath,amssymb,bm}
\usepackage{stmaryrd}
\usepackage{xspace,latexsym,epic,eepic,graphicx}
\usepackage{scalefnt}
\usepackage{stmaryrd}
\usepackage{subfigure}

\usepackage{tikz}
\usetikzlibrary{automata}
\usetikzlibrary{chains}
\usetikzlibrary{arrows}
\usetikzlibrary{calc}

\usepackage{mathabx}

 \usepackage{relsize}
 
 \usepackage{url}
 
 \usepackage{alltt}

\renewcommand{\epsilon}{\varepsilon}
\renewcommand{\phi}{\varphi}

\newcommand{\np}{{\sc NP}}

\newcommand{\ptime}{{\sc Ptime}}
\newcommand{\exptime}{{\sc Exptime}}
\newcommand{\twoexptime}{{\sc 2Exptime}}
\newcommand{\pspace}{{\sc Pspace}}

\newcommand{\expspace}{{\sc Expspace}}

\newcommand{\sem}[1]{\llbracket #1 \rrbracket}

\newtheorem{example}{Example}
\newtheorem{theorem}{Theorem} 
\newtheorem{proposition}{Proposition}
  
\newtheorem{definition}{Definition}

\newtheorem{lemma}{Lemma}

\newcommand{\beforeverbatim}{\vspace*{-4pt}}
\newcommand{\afterverbatim}{\vspace*{-7pt}}

\newcommand{\nodetests}{\textit{NodeTests}}
\newcommand{\bx}{\raisebox{0pt}{$\mathlarger{{\Box}}$}}
\newcommand{\dm}{\raisebox{-2pt}{$\mathlarger{\mathlarger{\mathlarger{\Diamond}}}\!$}}

\newcommand{\isobject}{\Obj}
\newcommand{\isarray}{\Arr}
\newcommand{\isstring}{\Str}
\newcommand{\isnumber}{\Int}

\newcommand{\minim}{\textit{Min}}
\newcommand{\maxim}{\textit{Max}}
\newcommand{\minch}{\textit{MinCh}}
\newcommand{\maxch}{\textit{MaxCh}}
\newcommand{\uniq}{\textit{Unique}}
\newcommand{\patternlogic}{\textit{Pattern}}
\newcommand{\mof}{\textit{MultOf}}
\newcommand{\compare}{\sim\!\!}

\newcommand{\type}{\texttt{"type":}}
\newcommand{\obj}{\texttt{"object"}}
\newcommand{\str}{\texttt{"string"}}
\newcommand{\num}{\texttt{"number"}}
\newcommand{\arr}{\texttt{"array"}}
\newcommand{\required}{\texttt{"required":}}
\newcommand{\addProp}{\texttt{"additionalProperties":}}
\newcommand{\patProp}{\texttt{"patternProperties":}}
\newcommand{\prop}{\texttt{"properties":}}
\newcommand{\items}{\texttt{"items":}}
\newcommand{\additionalItems}{\texttt{"additionalItems":}}

\newcommand{\uniqueItems}{\texttt{"uniqueItems":}}

\newcommand{\true}{\text{true}}

\newcommand{\pattern}{\texttt{"pattern":}}
\newcommand{\regexp}{\texttt{"}\textit{regexp}\texttt{"}}
\newcommand{\maxi}{\texttt{"maximum":}}
\newcommand{\mini}{\texttt{"minimum":}}
\newcommand{\multipleOf}{\texttt{"multipleOf":}}

\newcommand{\unfold}{\textit{unfold}}

\title{JSON: data model, query languages and \\schema specification}

\numberofauthors{4}

\author{\alignauthor
Pierre Bourhis  \\
\affaddr{CNRS CRIStAL UMR9189}\\
\and
\alignauthor
Juan L. Reutter \\
\affaddr{PUC Chile and Center for Semantic Web Research}\\
\and
\alignauthor
Fernando Su\'arez\\
\affaddr{PUC Chile and Center for Semantic Web Research}\\
\and
\alignauthor
Domagoj Vrgo\v{c} \\
\affaddr{PUC Chile and Center for Semantic Web Research}\\
}

\begin{document}
\maketitle

\sloppy

\begin{abstract}
Despite the fact that JSON is currently one of the most popular formats for exchanging data on the Web, there are very few studies on this topic and there are no agreement upon theoretical framework for dealing with JSON. Therefore in this paper we propose a formal data model for JSON documents and, based on the common features present in available systems using JSON, we define a lightweight query language allowing us to navigate through JSON documents. We also introduce a logic capturing the schema proposal for JSON and study the complexity of basic computational tasks associated with these two formalisms.
\end{abstract}

\section{Introduction}\label{sec:intro}

JavaScript Object Notation (JSON) \cite{json-ietf,json-ecma} is a lightweight format based on the data types of the JavaScript programming language. In their essence, JSON documents are dictionaries consisting of key-value pairs, where the value can again be a JSON document, thus allowing an arbitrary level of nesting. An example of a JSON document is given in Figure \ref{fig-json}. As we can see here, apart from simple dictionaries, JSON also supports arrays and atomic types such as integers and strings. Arrays and dictionaries can again contain arbitrary JSON documents, thus making the format fully compositional.

\begin{figure}[h]
\small
\begin{verbatim}
   { 
     "name": {
        "first": "John",
        "last": "Doe" 
        },
     "age": 32,
     "hobbies": ["fishing","yoga"]
   }
\end{verbatim}
\vspace*{-10pt}
\caption{A simple JSON document.}
\label{fig-json}
\end{figure}

Due to its simplicity, and the fact that it is easily readable both by humans and by machines, JSON is quickly becoming one of the most popular formats for exchanging data on the Web. This is particularly evident with Web services communicating with their users through an Application Programming Interface (API), as JSON is currently the predominant format for sending API requests and responses over the HTTP protocol. Additionally, JSON format is much used in database systems built around the NoSQL paradigm (see e.g. \cite{mongoDB,Arango,Orient}), or graph databases (see e.g. \cite{neo4j}).

Despite its popularity, the coverage of the specifics of JSON format in the research literature is very sparse, and to the best of our knowledge, there is still no agreement on the correct data model for JSON, no formalisation of the core query features which JSON systems should support, nor a logical foundation for JSON Schema specification. And while some preliminary studies do exist \cite{nurseitov2009comparison,LiuHM14,BCCRX16,PRSUV16}, as far as we are aware, no attempt to describe a theoretical basis for JSON has been made by the research community. Therefore, the main objective of this paper is to formally define an appropriate data model for JSON, identify the key querying features provided by the existing JSON systems, and to propose a logic allowing us to specify schema constraints for JSON documents.

In order to define the data model, we examine the key characteristics of JSON documents and how are they used in practice. As a result we obtain a tree-shaped structure very similar to the ordered data-tree model of XML \cite{bojanczyk2009two}, but with some key differences. 
The first difference is that JSON trees are deterministic by design, as each key can appear at most once inside a dictionary. This has various implications at the time of querying JSON documents: on one hand we sometimes deal with languages far simpler than XML, but on the other hand this key 
restriction can make static analysis more complicated, even for simpler queries. 
Next, arrays are explicitly present in JSON, which is not the case in XML. Of course, the ordered structure of XML could be used to simulate arrays, but the defining feature of each JSON dictionary is that it is unordered, thus dictating the nodes of our tree to be typed accordingly. 
And finally, JSON values are again JSON objects, thus making equality comparisons much more complex than in case of XML, since we are now comparing subtrees, and not just atomic values.
We cover all of these features of JSON in more detail in the paper, and  we also argue that, while technically possible (albeit, in a very awkward manner), coding JSON documents using XML might not be the best solution in practice. 

Next, we consider the problem of querying JSON. As there is no agreement upon query language in place, we examine an array of practical JSON systems, ranging from programming languages such as Python \cite{python}, to fully operational JSON databases such as MongoDB \cite{mongoDB}, and isolate what we consider to be key concepts for accessing JSON documents. As we will see, the main focus in many systems is on navigating the structure of a JSON tree, therefore we propose a navigational logic for JSON documents based on similar approaches from the realm of XML \cite{Figueira10}, or graph databases \cite{BPR12,LMV16}. 
We then show how our logic captures common use cases for JSON, extend it with additional features, and demonstrate that it respects the ``lightweight nature" of the JSON format, since it can be evaluated very efficiently, and it also has reasonable complexity of main static tasks. Interestingly, 
sometimes we can reuse results devised for other similar languages such as XPath or Propositional Dynamic Logic, but the nature of JSON and the functionalities present in query languages also demand new approaches or a refinement of these techniques. 

Another important aspect of working with any data format is being able to specify the structure of documents. A usual way to do this is through schema specification, and in the case of JSON, there is indeed a draft proposal for a schema language \cite{jsonschema} (called JSON Schema), which has recently been formalised in \cite{PRSUV16}. Based on the formalisation of \cite{PRSUV16} we define a logic capturing the full formal specification of JSON Schema, show that it is essentially equivalent to the navigational language we propose for querying JSON, and study the complexity of its evaluation and static tasks. However, once we take into account the recursive functionalities of JSON Schema, we arrive at a powerful new formalism that is much more difficult to encompass inside well known frameworks. 

Finally, since theoretical study of JSON is still in its early stages, we close with a series of open problems and directions for future research.

\smallskip
\noindent{\bf Organisation.} We formally define JSON and some of its features in Section \ref{sec:prelim}. The appropriate data model for JSON is discussed in Section \ref{sec:model}, and its query language in Section \ref{sec:nav}. In Section \ref{sec:schema} we define a logic capturing a schema specification for JSON. Our conclusions and the directions for future work are discussed in Section \ref{sec:future}. 
Due to the lack of space most proofs are placed in the appendix to this paper.

\section{Preliminaries}\label{sec:prelim}

\noindent{\bf JSON documents.} 
We start by fixing some notation regarding JSON documents. 
The full JSON specification defines seven types of values: objects, arrays, strings, numbers and 
the values true, false and null \cite{json}. However, to abstract from encoding details we assume 
our JSON documents are only formed by objects, arrays, strings and natural numbers. 

Formally, denote by $\Sigma$ the set of all unicode characters. JSON values are defined as follows. 
 First, any natural number $n \geq 0$ is a JSON value, called a \emph{number}. 
Furthermore, if \texttt{s} is a string in $\Sigma^*$, 
then \texttt{"s"} is a JSON value, called a \emph{string}. 
Next, if $v_1,\ldots,v_n$ are JSON values and $s_1,\ldots,s_n$ are pairwise distinct string values, then 
	$\{s_1:v_1,\ldots,s_n:v_n\}$ is a JSON value, called an \emph{object}. 
	In this case, each $s_i:v_i$ is called a key-value pair of this object.
Finally, if $v_1,\ldots,v_n$ are JSON values then $[v_1,\ldots,v_n]$ is a JSON value called an 
{\em array}. In this case $v_1,\ldots,v_n$ are called the \emph{elements} of the array. 
Note that in the case of arrays and objects the values $v_i$ can again be objects or arrays, thus allowing the documents an arbitrary level of nesting.


\smallskip
\noindent{\bf JSON navigation instructions.}
Arguably all systems using JSON base the extraction of information in what we call \emph{JSON navigation instructions}. 
The notation used to specify JSON navigation instructions varies from system to system, but it always follows the same 
two principles: 
\begin{itemize}\itemsep=0pt
\item If ${J}$ is a JSON object, then one should be able to access the JSON value in a specific key-value pair of this object. 
\item If ${J}$ is a JSON array, then one should be able to access the $i$-th element of ${J}$. 
\end{itemize}

In this paper we adopt the python notation for navigation instructions: If ${J}$ is an object, then 
${J[\textit{key}]}$ is the value of ${J}$ whose key is the string value \texttt{"key"}. Likewise, if ${J}$ is an array, then 
${J[n]}$, for a natural number $n$, contains the n-th element of ${J}$\footnote{Some JSON systems 
prefer using a \emph{dot} notation, where ${J.\text{key}}$ and $\text{J.n}$ are the equivalents of ${J[\text{key}]}$ and 
${J[n]}$.}. 

As far as we are aware, all JSON systems use JSON navigation instructions as a primitive for querying JSON documents, and in particular 
it is so for the systems we have reviewed in detail: Python and other programming languages \cite{python}, the mongoDB database \cite{mongoDB}, 
the JSON Path query language \cite{jpath} and the SQL++ project that tries to bridge relational and JSON databases \cite{sqlpp}. 

At this point it is important to note a few important consequences of using JSON navigation instructions, as 
these have an important role at the time of formalising this framework. 

First, note that we do not have a way of obtaining the keys of JSON objects. For example, if $\texttt{J}$ is the object 
\texttt{\{"first":"John", "last":"Doe"\}} we could issue the instructions ${J[\textit{first}]}$ to obtain the value of the pair 
\texttt{"first":"John"}, which is the string value \texttt{"John"}. Or use ${J[\textit{last}]}$ to obtain the value of the pair \texttt{"last":"Doe"}. However, 
 there is no instruction 
that can retrieve the keys inside this document (i.e. \texttt{"first"} and \texttt{"last"} in this case). 

Similarly, for the case of the arrays, the access is essentially a \emph{random access}: we can access the $i$-th element of an array, 
and most of the time there are primitives to access the first or the last element of arrays. However, we can not reason about different elements of the array. For example, for the array $K = \texttt{[12,5,22]}$ we can not retrieve, say, 
an element (or any element) which is greater than the first element of $K$. 
Some systems do feature FLWR expressions that support iterating over all elements. 
But this iteration is itself treated as a series of random accesses, using commands such as \texttt{For i in (0,n) print(J[i])}. 

\section{Data model for JSON}\label{sec:model}
\newcommand{\Achild}{\mathcal A} 
\newcommand{\Ochild}{\mathcal O} 
\newcommand{\values}{\textit{val}} 
\newcommand{\Arr}{\textit{Arr}}
\newcommand{\Obj}{\textit{Obj}}
\newcommand{\Str}{\textit{Str}}
\newcommand{\Int}{\textit{Int}}
\newcommand{\naturals}{\mathbb{N}}
\newcommand{\json}{\textit{json}}
\newcommand{\child}{X}
\newcommand{\EQ}{\textit{EQ}}

In this section we propose a formal data model for JSON documents whose goal is to closely reflect the manner in which JSON is manipulated using JSON navigation instructions, 
and that will be used later on as the basis of our formalisation of JSON query and schema languages.  
We begin by introducing our formal model, called JSON trees. 
 Afterwards we discuss the main differences between JSON trees and other well-studied tree formalisms such as data trees or XML. 



\subsection{JSON trees}\label{ss-json_trees}

JSON objects 
are by definition compositional: each JSON object is a set of key-value pairs, in which values can again be JSON objects. 
This naturally suggests using a tree-shaped structure to model JSON documents. 
However, this structure must preserve the compositional nature of JSON. That is, if each node of the tree structure represents a JSON document, 
then the children of each node must represent the documents nested within it. 
For instance, consider the following JSON document $J$.

\beforeverbatim
{\footnotesize
\begin{verbatim}
   { 
     "name": {
        "first": "John",
        "last": "Doe" 
        },
     "age": 32
   }
\end{verbatim}
}
\afterverbatim

\noindent as explained before, this document is a JSON {\em object} which contains two keys: \texttt{"name"} and \texttt{"age"}. Furthermore, the value of the key \texttt{"name"} is another JSON document 
and the value of the key \texttt{"age"} is the integer 32. There are in total 5 JSON values inside this object: the complete document itself, plus the literals \texttt{32}, \texttt{"John"} and \texttt{"Doe"}, and the object \texttt{"name": \{"first":"John", "last":"Doe"\}}. 
So how should a tree representation of the document $J$ look like? If we are to preserve the compositional structure of JSON, then the most natural 
representation is by using the following edge-labelled tree:

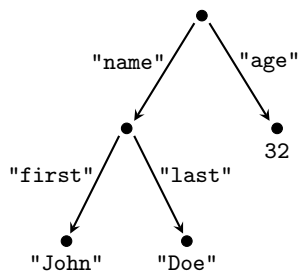
\begin{center}
\begin{tikzpicture}[>=stealth]
\filldraw [black] (0,0) circle (2pt)
(-1,-1.5) circle (2pt)
(1,-1.5) circle (2pt)
(-1.8,-3) circle (2pt)
(-0.2,-3) circle (2pt)
;


\path (-0.1,-0.1) edge[->, thick] node[pos=0.4,left=-2pt] {\small \texttt{"name"}} (-0.9,-1.4);
\path (0.1,-0.1) edge[->, thick] node[pos=0.4,right=-2pt] {\small \texttt{"age"}}(0.9,-1.4);
\path (-1.1,-1.6) edge[->, thick] node[pos=0.4,left=-2pt] {\small \texttt{"first"}} (-1.75,-2.9);
\path (-0.9,-1.6) edge[->, thick] node[pos=0.4,right=-2pt] {\small \texttt{"last"}} (-0.25,-2.9);


\draw (-1.8,-3) node[below=2pt] {\small \texttt{"John"}};
\draw (-0.2,-3) node[below=2pt] {\small \texttt{"Doe"}};
\draw (1,-1.5) node[below=2pt] {\small \texttt{32}};

\end{tikzpicture}
\end{center}

The root of tree represents the entire document. The two edges labelled \texttt{"name"} and \texttt{"age"} represent two keys inside this JSON object, and they lead to 
nodes representing their respective values. In the case of the key \texttt{"age"} this is just an integer, while in the case of \texttt{"name"} we obtain another JSON object that is represented as a subtree of the entire tree. 

Finally, we need to enforce the property of JSON that no object can have two keys with the same name, thus making the model deterministic in some sense, since each node will have only one child reachable by an edge with a specific label. 
Let us briefly summarise the properties of our model so far. 


\smallskip
\noindent
\emph{Labelled edges}. Edges in our model are labelled by the keys forming the key-value pairs of objects. This means that we can directly follow 
the label of edges when issuing JSON navigation instructions, and also means that information of keys is represented in a different 
medium than JSON values (labels for the former, nodes for the latter). 
This is inline with the way JSON navigation instructions work, as one can only retrieve values of key-value pairs, but not the keys themselves. 
%
To comply with the JSON standard, we disallow trees where a same edge label is repeated in two different edges leaving a node. 

\smallskip
\noindent
\emph{Compositional structure.} One of the advantages of our tree representation is that any of its subtrees represent a JSON document themselves. In fact, the five possible subtrees of the tree above correspond to the five JSON values present in the JSON $J$. 

\smallskip
\noindent
\emph{Atomic values.} Finally, some elements of a JSON document are actual values, such as integers or strings. For this reason leaf nodes corresponding to integers and strings will also be assigned a value they carry. Leaf nodes without a value represent empty objects: that is, documents of the form \texttt{\{\}}.

\smallskip

Although this model is simple and conceptually clear, we are missing a way of representing arrays. 
Indeed, consider again the document from Figure \ref{fig-json} (call this document $J_2$).
%
%
%
In $J_2$ the value of the key \texttt{"hobbies"} is an array: another feature explicitly present in JSON that thus needs to be reflected in our model. 

As arrays are ordered, this might suggest that we can have some nodes whose children form an ordered list of siblings, much like in the case of XML. 
But this would not be conceptually correct, for the following two reasons. 
First, as we have explained, JSON navigation instructions use random access to access elements in arrays. 
For example, the navigation instruction used to retrieve an element of an array is of the form ${J_2[\textit{hobbies}][i]}$, aimed at obtaining the $i$-th element of the array 
under the key \texttt{"hobbies"}. 
But more importantly, we do not want to treat arrays as a list because lists naturally suggest navigating through different elements of the list. On the contrary, 
none of the systems we reviewed feature a way of navigating form one element of the array to another element. That is,  
once we retrieve the first element of the array under the key \texttt{"hobbies"}, we have no way of linking it to its siblings.

We choose to model JSON arrays as nodes whose children are accessed by axes labelled with natural numbers reflecting their position in the array. 
Namely, in the case of JSON document $J_2$ above we obtain the following representation:

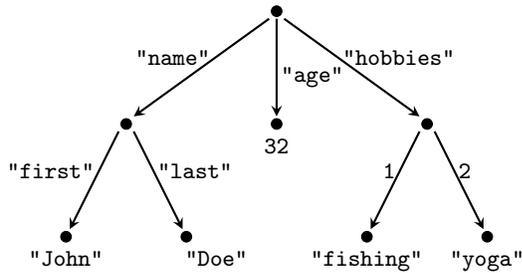
\begin{center}
\begin{tikzpicture}[>=stealth]
\filldraw [black] (0,0) circle (2pt)
(-2,-1.5) circle (2pt)
(0,-1.5) circle (2pt)
(-2.8,-3) circle (2pt)
(-1.2,-3) circle (2pt)

(2,-1.5) circle (2pt)
(1.2,-3) circle (2pt)
(2.8,-3) circle (2pt)
;


\path (-0.1,-0.1) edge[->, thick] node[pos=0.4,left=-2pt] {\small \texttt{"name"}} (-1.9,-1.4);
\path (0,-0.1) edge[->, thick] node[pos=0.6,right=-2pt] {\small \texttt{"age"}}(0,-1.4);
\path (-2.1,-1.6) edge[->, thick] node[pos=0.4,left=-2pt] {\small \texttt{"first"}} (-2.75,-2.9);
\path (-1.9,-1.6) edge[->, thick] node[pos=0.4,right=-2pt] {\small \texttt{"last"}} (-1.25,-2.9);

\path (0.1,-0.1) edge[->, thick] node[pos=0.4,right=-2pt] {\small \texttt{"hobbies"}}(1.9,-1.4);
\path (1.9,-1.6) edge[->, thick] node[pos=0.4,left=-2pt] {\small \texttt{1}}(1.25,-2.9);
\path (2.1,-1.6) edge[->, thick] node[pos=0.4,right=-2pt] {\small \texttt{2}}(2.75,-2.9);


\draw (-2.8,-3) node[below=2pt] {\small \texttt{"John"}};
\draw (-0.8,-3) node[below=2pt] {\small \texttt{"Doe"}};
\draw (0,-1.5) node[below=2pt] {\small \texttt{32}};

\draw (1.2,-3) node[below=2pt] {\small \texttt{"fishing"}};
\draw (2.8,-3) node[below=2pt] {\small \texttt{"yoga"}};

\end{tikzpicture}
\end{center}

Having arrays defined in this way allows us still to treat the child edges of our tree as navigational axes: before we used a key such as \texttt{"age"} to traverse an edge, and now we use the number labelling the edge to traverse it and arrive at the child. 


\smallskip
\noindent
\textbf{Formal definition}. 
As our model is a tree, we will use tree domains as its base. A {\em tree domain} is a prefix-closed subset of $\naturals^*$. Without loss of generality we assume that for 
all tree domains $D$, if $D$ contains a node $n \cdot i$, for $n \in \naturals^*$ then $D$ contains all $n \cdot j$ with $0 \leq j < i$.

Let $\Sigma$ be an alphabet. A {\bf JSON tree} over $\Sigma$ is a structure 
$J = (D,\Obj,\Arr,\Str,\Int,\Achild,\Ochild,\values)$, where 
D is a tree domain that is partitioned by $\Obj$, $\Arr$, $\Str$ and $\Int$, 
$\Ochild\subseteq \Obj \times \Sigma^* \times D$ is the object-child relation, 
$\Achild\subseteq \Arr \times \naturals \times D$ is the array-child relation, 
$\values: \Str \cup \Int \rightarrow \Sigma^* \cup \naturals$ is the string and number \emph{value} function, and 
where the following holds: 

\begin{itemize}\itemsep=0pt
\item[1] For each node $n \in \Obj$ and child $n \cdot i$ of $n$, $\Ochild$ contains one triple $(n,w,n\cdot i)$, for a word $w \in \Sigma^*$.
\item[2] The first two components of $\Ochild$ form a \emph{key}: if $(n,w,n\cdot i)$ and $(n,w,n\cdot j)$ 
are in $\Ochild$, then $i = j$. 
\item[3] For each node $n \in \Arr$ and child $n \cdot i$ of $n$, $\Achild$ contains the triple $(n,i,n\cdot i)$. 
\item[4] If $n$ is in $\Str$ or $\Int$ then $D$ cannot contain nodes of form $n \cdot u$. 
\item[5] The value function assigns to each string node in $\Str$ a value in $\Sigma^*$ and to each number node in $\Int$ a natural number,
\end{itemize}

The usage of a tree domain is standard, and we have elected to explicitly partition the domain into four types of nodes: 
$\Obj$ for objects, $\Arr$ for arrays, $\Str$ for strings and $\Int$ for integers. 
The first and second conditions specify that edges between objects and their children are labelled with words, but 
we can only use each label one time per each node. The third condition specifies that the edges between arrays and their children are labelled with the number representing the order of children. 
The fourth condition simply states that strings and numbers must be leaves in our trees, and the fifth condition describes 
the value function $\values$. 
Note that we have explicitly distinguished the four type of JSON documents (objects, arrays, strings and integers). 
This is important when modelling schema definitions for JSON, as we shall show later. 


Throughout this paper we will use the term JSON tree and JSON interchangeably. 
As already mentioned above, one important feature of our model is that when looking at any node of the tree, a subtree rooted at this node is again a valid JSON. We can therefore define, for a JSON tree $J$ and a node $n$ in $J$, a function $\json(n)$ which returns the subtree of $J$ rooted at $n$. Since this subtree is again a JSON tree, the value of  $\json(n)$ is always a valid JSON. 

\subsection{JSON and XML}

Before continuing we give a few remarks about differences and similarities between JSON and XML, and how are these reflected in their underlying data models. We start by summarising the differences between the two formats.
%
\begin{enumerate}\itemsep=0pt
\item {\em JSON mixes ordered and unordered data}. JSON objects are completely without order, but for arrays we can do 
random access depending on their position.
On the other hand, 
XML enforces a strict order between the children of each node. Coding JSON as XML would imply permitting sibling traversal 
for some nodes, but disallowing it for others. We can do that with XML with some ad-hoc rules, but this is precisely what we do in our model in a much cleaner way.  
\item {\em JSON Arrays are neither lists nor sets}. As we have explained, we have random access, but we do not have the possibility 
of sibling traversal. Enforcing this in languages such as XPath is a very cumbersome task. 
\item {\em JSON trees are deterministic.} The property of JSON tree which imposes that all  keys of each object have to  be distinct makes JSON trees deterministic in the sense that if we have the key name, there can be at most one node reachable through an edge labelled with this key. On the other hand, XML trees are nondeterministic since there are no labels on the edges, and a node can have multiple children. As we will see, the deterministic nature of JSON can make some problems more difficult than in the XML setting.
\item {\em Value is not just in the node, but is the entire subtree rooted at that node.} Another fundamental difference is that in XML when we talk about values we normally refer to the value of an attribute in a node. On the contrary, it is common for systems using JSON to allow comparisons of the full subtree of a node with a nested JSON document, or even comparing two nodes themselves in terms of their subtrees. To be fair, in XML one could also argue this to be true, but unlike in the case of XML, these ``structural" comparisons are intrinsic in most JSON languages, as we discuss in the following sections.
\end{enumerate}

On the other hand, it is certainly possible to code JSON documents using the XML data format. 
In fact, the model of ordered unranked trees with labels and attributes, which serves as the base of XML, 
was shown to be powerful enough to code some very expressive database formats, such as relational and even graph data. 
However, both models have enough differences to justify a study of JSON on its own. This is particularly evident when considering navigation through JSON documents, where keys in each object have to be unique, thus allowing us to obtain values very efficiently. On the other hand, coding JSON as XML would require us to have keys as node labels, thus forcing a scan of all of the node's children in order to retrieve the value.

\section{Navigational queries over JSON}\label{sec:nav}

As JSON navigation instructions are too basic to serve as a complete query language,  
most systems have developed different ways of querying JSON documents. 
Unfortunately, there is 
no standard, nor general guidelines, about how documents are accessed. As a result the syntax and operations 
between systems vary so much that it would be almost impossible to compare them. 
Hence, it would be desirable to identify a common core of functionalities shared between these systems, or at least a general picture of how such query languages look like. Therefore we begin this section by reviewing the most common operations available in current JSON systems.

Here we mainly focus on the subdocument selecting functionalities of JSON query languages. By subdocument selecting we mean 
functionalities that are capable of finding or highlighting specific parts within JSON documents, either to be returned immediately or to 
be combined as new JSON documents. 
As our work is not intended to be a survey, we have not reviewed all possible systems available to date. 
However, we take inspiration from MongoDB's query language (which arguably has served as a basis for 
many other systems as well, see e.g. \cite{Arango,Orient,CouchDB}), and JSONPath \cite{jpath}and SQL++ \cite{sqlpp}, two other query languages that have been proposed by the community.

Based on this, we propose a navigational logic that can serve as a common core to define a standard way of 
querying JSON. We then define several extensions of this logic, such as allowing nondeterminism or recursion, and study how these affect basic reasoning task such as evaluation and containment.
\subsection{Accessing documents in JSON databases}

Here we briefly describe how JSON systems query documents. 

\smallskip
\noindent
\textbf{Query languages inspired by FLWR or relational expressions}. 
There are several proposals to construct query languages that can merge, join and even produce new JSON documents. 
Most of them are inspired either by XQuery (such as JSONiq \cite{jsoniq}) or SQL (such as SQL++ \cite{sqlpp}). These languages have of course 
a lot of intricate features, and to our best extent have not been formally studied. However, in terms of JSON navigation thy all seem to support basic JSON navigation instructions and not much more. 

\smallskip
\noindent
\textbf{MongoDB's find function}.  
The basic querying mechanism of MongoDB is given by the \emph{find} function \cite{mongoDB}, therefore we focus on this aspect of the system\footnote{For a detailed study of other functionalities MongoDB offers see e.g. \cite{BCCRX16}. Note that this work does not consider the find function though.}. 
The find function receives two parameters, which are both JSON documents. Given a collection of JSON document and these parameters, the find function then
produces an array of JSON documents.

The first parameter of the find function serves as a \emph{filter}, and its goal is to select some of the JSON documents from the input. The second parameter is the \emph{projection}, and as its
name suggests, is used to specify which parts of the filtered documents are to be returned. Since our goal is specifying a navigational logic, we will only focus 
on the filter parameter, and on find queries that only specify the filter. We return to the projection in Section \ref{sec:future}. For more details we refer the reader to the current version of the documentation \cite{mongoDB}.
 

The basic building block of filters are what we call \emph{navigation condition}, which can be visualised as expressions of the form $P \sim J$, where 
$P$ is a JSON navigation instruction, $\sim$ is a comparison operator (MongoDB allows all the usual $<$, $\leq$, $=$ , $\geq$, $>$, and several others operators) and 
$J$ is a JSON document. 

\begin{example} Assume that we are dealing with a collection of JSON files containing information about people and that we want to obtain the one describing a person named Sue. In MongoDB this can be achieved using the following query \verb+db.collection.find({name: {$eq: "Sue"}},{})+. The initial part \verb+db.collection+ is a system path to find the collection of JSON documents we want to query. Next, \verb+"name"+ is a simple navigation instruction 
used to retrieve the value under the key \verb+"name"+. Last, the expression \verb+{$eq: "Sue"}+ is used to state that the JSON document retrieved 
by the navigation instruction is equal to the JSON \verb+"Sue"+. Since we are not dealing with projection, the second parameter is simply the empty document \verb+{}+. Using the notation above we could also write this navigation condition as $J[\textit{name}] = $\verb+"Sue"+.
\end{example}


%

Finally, navigation conditions can be combined using boolean operations with the standard meaning. Also note that filters always return entire documents. If we want a part of a JSON file we need to use filters.
  
\smallskip
\noindent
\textbf{Query languages inspired by XPath or relational expressions}. 
The languages we analysed thus far offer very simple navigational features. However, people also recognized the need to allow more complex properties such as nondetermnistic navigation, expression filters and allowing arbitrary depth nesting through recursion. As a result, an adaptation of the XML query language XPath to the context of JSON, called JSONPath \cite{jpath} was introduced and implemented (see e.g. \url{https://github.com/jayway/JsonPath}).

\smallskip

Based on these features, we first introduce a logic capturing basic queries provided by navigation instructions and conditions, and then extend it with non-determinism and recursion resulting in a logic  resembling similar approaches over XML.
%
%

\subsection{Deterministic JSON logic}

The first logic we introduce is meant to capture JSON navigation instructions and other deterministic forms of querying such as 
MongoDB's find function. We call this logic \emph{JSON navigation logic}, or JNL for short. 
We believe that this logic, although not very powerful, is interesting in its own right, as it leads to very lightweight 
algorithms and implementations, which is one of the aims of the JSON data format. 

As often done in XML\cite{Figueira10} and graph data\cite{LMV16}, we define ours in terms of unary and binary formulas. 
\begin{definition}[JSON navigational logic]
Unary formulas $\varphi, \psi$  
and binary formulas $\alpha, \beta$ of the \emph{JSON navigational logic} are expressions satisfying the grammar
\begin{equation*}
\def\arraystretch{1.2}
\begin{array}{lll}
\alpha,\beta  & :=  & \langle \phi \rangle \ |\ \child_{w}  \ |\ \child_{i}  \ |\ \alpha \circ \beta \ |\ \epsilon \\
\varphi,\psi & :=  & \top  \ |\ \lnot \varphi\ 
|\ \varphi \land \psi\ |\ \phi\vee \psi\ |\ [\alpha]\ |\\
 & & \hspace*{70pt} \EQ(\alpha,A)\ |\ \EQ(\alpha,\beta)
 \end{array}
\label{detjl-syntax}
\end{equation*}
where $w$ is a word in $\Sigma^*$, $i$ is a natural number and $A$ is an arbitrary JSON document. 
\end{definition}

Intuitively, binary operators allow us to move through the document (they connect two nodes of a JSON tree), and unary formulas check whether a property is true at some point of our tree. For instance, $\child_{w}$ and $\child_{i}$ allow basic navigation by accessing the the value of the key named $w$, or the $i$th element of an array respectively. They can subsequently be combined using composition or boolean operations to form more complex navigation expressions. Unary formulas serve as tests if some property holds at the part of the document we are currently reading. These also include the operator $[\alpha]$ allowing us to test if some binary condition is true starting at a current node (similarly, $\langle \varphi \rangle$ allows us to combine node tests with navigation). Finally, the comparison operators  $\EQ(\alpha,A)$ and $\EQ(\alpha,\beta)$ simulate XPath style tests which check whether a current node can reach a node whose value is $A$, or if two paths can reach nodes with the same value. The difference from XML though, is that this value is again a JSON document and thus a subtree of the original tree.

The semantics of binary formulas is given by the relation 
$\sem{\alpha}_{J}$, for a binary formula $\alpha$ and a JSON $J$, and it selects pairs of nodes of $J$:  
\begin{itemize}
\item $\sem{\langle\phi\rangle}_{J} = \{(n,n) \mid n \in \sem{\phi}_J\}$.
\item $\sem{\child_w}_{J} = \{(n,n') \mid (n,w,n')\in \Ochild\}$.
\item $\sem{\child_{i}}_{J} = \{(n,n') \mid (n,i,n') \in \Achild \}$, for $i \in \naturals$.
\item $\sem{\alpha \circ \beta}_{J} = \sem{\alpha}_{J} \circ \sem{\beta}_{J}$.
\item $\sem{\epsilon}_{J} = \{(n,n) \mid n$ is a node in $J\}$.
\end{itemize}
For the semantic of the unary operators, let us assume that $D$ is the domain of $J$. 
\begin{itemize}
\item $\sem{\top}_J = D$. 
\item $\sem{\neg \varphi}_J = D - \sem{\phi}_J$. 
\item $\sem{\varphi \wedge \psi}_J = \sem{\phi}_J \cap \sem{\psi}_J$. 
\item $\sem{\varphi \vee \psi}_J = \sem{\phi}_J \cup \sem{\psi}_J$. 
\item $\sem{[\alpha]}_J = \{n \mid n \in D$ and there is a node $n'$ in $D$ such that $(n,n') \in \sem{\alpha}_{J}\}$
\item $\sem{\EQ(\alpha,A)}_J = \{n \mid n \in D$ and there is a node $n_1$ in $T$ such that $(n,n_1) \in \sem{\alpha}_{J}$ and $\json(n_1) = A\}$ 
\item $\sem{\EQ(\alpha,\beta)}_J = \{n \mid n \in D$ and there are nodes $n_1,n_2$ in $T$ such that $(n,n_1) \in \sem{\alpha}_{J}$, 
$(n,n_2) \in \sem{\beta}_{J}$, and $\json(n_1) = \json(n_2)\}$. 
\end{itemize}

Typically, most systems allow jumping to the last element of an array, or the the $j$-th element counting from the last to the first. 
To simulate this we can allow binary expressions of the form $\child_{i}$, for an integer $i <0$, where $-1$ states the last position of the array, and 
$-j$ states the $j$-th position starting from the last to the first. Having this dual operator would not change any of our results, but 
we prefer to leave it out for the sake of readability. 

\smallskip
\noindent
\textbf{Algorithmic properties of JNL}. As promised, here we show that
JNL is a logic  particularly well behaved for database applications. For this we study the evaluation problem and satisfiability problem associated with JNL. 
The \textsc{Evaluation} problem asks, on input a JSON $J$, a JNL unary expression $\phi$ and a node 
$n$ of $J$, whether $n$ is in $\sem{\phi}_J$. The \textsc{Satisfiability} problem asks, on input a JNL expression 
$\phi$ , whether there exists a JSON $J$ such that $\sem{\phi}_J$ is nonempty. 
We start with evaluation, showing that JNL indeed matches the ``lightweight" spirit of the JSON format and can be evaluated very efficiently:

\begin{proposition}
\label{prop-JNL-eval}
The \textsc{Evaluation} problem for JNL 
can be solved in time $O(|J|\cdot|\phi|)$.
\end{proposition}

For this result, we can reuse techniques for XPath evaluation (see e.g. \cite{Parys09,GottlobKP05}). However, the 
presence of the $\EQ(\alpha,\beta)$ operator forces us to refine these techniques in a non-trivial way.
A straightforward way of incorporating this predicate into XPath algorithms is 
to pre-process all pairs of nodes to see which pairs have equal subtrees, but this only gives us a quadratic algorithm. 
Instead, we transform our JNL formula into an equivalent non recursive monadic datalog program with stratified negation \cite{Gottlob2006},
and show how to evaluate the latter by doing equality comparisons ``online" as they appear.

Next, we move to satisfiability, showing that the complexity of the problem is best possible, considering that JNL can emulate propositional formulas.

\begin{proposition}
\label{prop-JNL-sat}
The \textsc{Satisfiability} problem for JNL is \np-complete. It is \np-hard even for formulas not using negation nor the equality operator.  
\end{proposition}

It might be somewhat surprising that the positive fragment without data comparisons is not trivially satisfiable. This holds due to the fact that each key in an object is unique, so a formula of the form $X_a[X_1]\wedge X_a[X_b]$ is unsatisfiable because it forces the value of the key $a$ to be both an array and a string at the same time.

\subsection{Extensions}

Although the base proposal for JNL captures the deterministic spirit of JSON, it is somewhat limited in expressive power. Here we propose two natural extensions: the ability to non-deterministically select which child of a node is selected, and the ability to traverse paths of arbitrary length.

\noindent
\textbf{Non-determinism}.
The path operators $ \child_{w}$ and $ \child_{i}$ can be easily extended such that they return more than a single child; namely, we can permit matching of regular expressions and intervals, instead of simple words and array positions.

Formally, non-deterministic JSON logic extends binary formulas of JNL by the following grammar: 
\begin{equation*}
\begin{array}{lll}
\alpha,\beta  & :=  & \langle \phi \rangle \ |\ \child_{e}  \ |\ \child_{i:j} \ |\ \alpha \circ \beta \ |\ \epsilon \\
 \end{array}
\label{nondetjl-syntax}
\end{equation*}
where $e$ is a subset of $\Sigma^*$ (given as a regular expression), and $i \leq j$ are natural numbers, or $j=+\infty$ (signifying that we want any element of the array following $i$).
 The semantics of the new path operators is as follows:
 
\begin{itemize}
\item $\sem{\child_e}_{J} = \{(n,n') \mid $ there is $w \in L(e)$ such that $(n,w,n')\in \Ochild\}$.
\item $\sem{\child_{i:j}}_{J} = \{(n,n') \mid $ there is $i \leq p \leq j$ such that the triple $(n,p,n')$ is in $\Achild \}$. 
\end{itemize}

\noindent
\textbf{Recursion}.
In order to allow exploring paths of arbitrary length we add the Kleene star to our logic. That is, recursive JNL allows $(\alpha)^*$ as a binary formula
(as usual we normally omit the brackets when the precedence of operators is clear). The semantics of 
$(\alpha)^*$ is given by 
$$\sem{(\alpha)^*} = \sem{\epsilon}_J \cup \sem{\alpha}_J \cup \sem{\alpha \circ \alpha} \cup \sem{\alpha \circ \alpha \circ \alpha}_J \cup \dots.$$ 

So what happens to the evaluation and containment when we extend this logic? 
For the case of evaluation, we can easily show that the linear algorithm is retained as long as we do not 
have the binary equality operator $\EQ(\alpha,\beta)$. Indeed, in this case, the evaluation can be done using the classical PDL model checking algorithm \cite{AI00,CleavelandS93} with small extensions which account for the specifics of the JSON format.
However, we are not able to extend the linear algorithm for the full case, because an expression of the form $\EQ(\alpha,\beta)$ might require checking all pairs of nodes in our tree for equality, resulting in a jump in complexity. 

\begin{proposition}
\label{prop-JNL-rec-eval}
The evaluation problem for JNL with non-determinism and recursion can be solved in time $O(|J|^3\cdot|\phi|)$, and in time $O(|J|\cdot|\phi|)$ if the formula does not use the predicate $\EQ(\alpha,\beta)$.
\end{proposition}

For satisfiability the situation is radically different, as the combination of recursion, non-determinism and 
the binary equalities ends up being too difficult to handle. 

\begin{proposition}
\label{prop-JNL-rec-sat-und}
The \textsc{Satisfiability} problem is undecidable for non-deterministic recursive JNL formulas, even if they do not use negation.
\end{proposition}

However, if we rule out the equality operator we can show much better bounds. For the full 
non-deterministic, recursive JNL (without equalities) the satisfiability problem is the same as 
other similar fragments such as PDL. For (non-recursive) non-deterministic JNL the problem is slightly easier. 

\begin{proposition}
\label{prop-JNL-rec-sat-dec}
The \textsc{Satisfiability} problem is: 
\begin{itemize}
\item \pspace-complete for non-deterministic, non-recursive JNL without the $\EQ(\alpha,\beta)$ operator. 
\item \exptime-complete for non-deterministic, recursive JNL without the $\EQ(\alpha,\beta)$ operator. 
\end{itemize}
\end{proposition}

Note that \pspace-hardness for satisfiability follows easily from the fact that we are now allow regular expressions in our 
edges: Given a regular expression $e$, we have that the $e$ is universal 
if and only if the query $[\child_{\Sigma^*}] \wedge \neg[\child_e]$ is not satisfiable. 
However, in the proof of this proposition we in fact show that the problem remains \pspace-hard even when 
the only regular expression which is not a word in a $\child_e$ axis is $\Sigma^*$. One can also show that \pspace-hardness remains 
when one only considers JSON documents without object values. 

\section{Schema definitions for JSON}
\label{sec:schema}

Having dealt with navigational primitives for querying JSON, our next task is to analyse JSON
Schema definitions. We focus solely on the JSON Schema specification \cite{json-ietf}, 
which is, up to our best knowledge, the only attempt to define a general schema 
language for JSON documents. The JSON Schema specification is currently in its fourth draft, and on its way of becoming 
an IETF standard. 

\subsection{JSON Schema}

As before, we first briefly present how JSON Schema works. We remark again that our intention is not to provide a full analysis 
for the specification, but rather show how the navigation works, with the aim of obtaining a logic that can capture JSON Schema. 
We thus concentrate on a core fragment that is equivalent to the full specification; we refer to \cite{PRSUV16} for more details and a full 
formalisation of this core. 

Every JSON schema is JSON document itself. JSON Schema can specify that a document must be any of the different types of 
values (objects, arrays, strings or numbers); and for each of these types there are several keywords that help shaping 
and restricting the set of documents that a schema specifies. The most important 
keyword is the "type" keyword, as it determines the type of value that has to be validated against the 
schema: a document of the form  $\{\type\str,\ \dots\}$ specifies string values, 
$\{\type\num,\ \dots\}$ specifies number values, $\{\type\obj,\ \dots\}$ specifies objects 
and $\{\type\arr,\ \dots\}$ specifies arrays. In addition to the type keyword, each schema includes a number of other 
pairs that shape the documents they describe. 

  \begin{table*}
 \begin{tabular}{c c }
\raisebox{0cm}{
\parbox{0.85\columnwidth}{
\fbox{\parbox{0.8\columnwidth}{

Keywords for string schemas:

\smallskip
- $\type\str$ \ \ \ \ - $\texttt{"pattern":}\ \exp $

}} \\ 

\smallskip

\fbox{\parbox{0.8\columnwidth}{
Keywords for number schemas:

\smallskip
- $\type\num$  \ \ \ \ -$ \texttt{"multipleOf":}\ {i} $ \\
- $ \texttt{"minimum":}\ {i} $ \ \ \ \ \ \ \ \ \ \ - $ \texttt{"maximum":}\ {i} $

}}}}
\smallskip
&
\fbox{\parbox{1.1\columnwidth}{
Keywords for object schemas:

\smallskip
- $\type\obj$ \ \ \ \ \ \ \ \ \ \ \ \ \ - $\texttt{"required":}\ [\ k_1,\dots,k_n ]$ \\ 
- $\texttt{"minProperties": {i}}$ \ \ \ \ \ \ \ - $\texttt{"maxProperties": {i}}$ \\
- $\prop\{k_1:J_1,\dots,k_m:J_m\}$ \\ 
- $\patProp\{\texttt{"}e_1\texttt{":}J_1,\dots,\texttt{"}e_\ell\texttt{":}J_\ell\}$ \\
- $\texttt{"additionalProperties":}\  J  $
}} \\
\parbox{0.85\columnwidth}{
\fbox{\parbox{0.8\columnwidth}{
Keywords for array schemas:

\smallskip

- $\items[J_1,\dots,J_n]$ \\
- $\uniqueItems\true$ \\
- $\additionalItems J$ 
}}} & 
\fbox{\parbox{1.1\columnwidth}{
Boolean combination and comparisons:

\smallskip

- $\texttt{"anyOf":}\ [\ J_1,\dots,J_n ]$ \ \ \ \ \ \ \ \ - $\texttt{"allOf":}\ [\ J_1,\dots,J_m]$ \\ 
- $\texttt{"not":}\ J$ \ \ \ \ \ \ \ \ \ \ \ \ \ \ \ \ \ \ \ \ \ \ \ \ \ - $\texttt{"enum":}\ [\ A_1,\dots,A_n ]$
}}
\end{tabular}
\caption{\textnormal{The form for all keyords in JSON schema. Here $i$ is always a natural number, $J$ and each $J_i$ are JSON schemas, $A_1,\dots,A_n$ are JSON documents, each $k_i$ is a string value ($k$ stand for key); and $\exp$ and each $\exp_i$ are regular expressions over the alphabet $\Sigma$ of strings. }}
\label{jsch-grammar}
\end{table*}

We now describe each of the four types of basic schemas. Table \ref{jsch-grammar} contains a list of all keywords available for each of these schemas. 

\medskip
\noindent 
\textbf{String schemas}. String schemas are those featuring the $\type\str$ pair. 
Additionally, they may include the pair $\pattern\regexp$, for $\textit{regexp}$ a regular expression over $\Sigma$, which 
validates only against those strings that belong to the language of this expression. 
For example, 
$\{\type\str\}$ and $\{\type\str,\pattern\texttt{"$(01)^+$"}\}$ are string schemas. The first schema 
validates against any string, and the second only against strings built from $0$ or $1$.

\medskip
\noindent 
\textbf{Number schemas}. 
For numbers, we can use the pair $\mini i$ to specify that the number is at least $i$, 
$\maxi i$ to specify that the number is at most $i$, and $\multipleOf i$ to specify 
that a number must be a multiple of $i$.  
Thus for example 
$\{\type\num,\maxi 12,\multipleOf 4\}$ describes numbers 0, 4, 8 and 12.

\medskip
\noindent 
\textbf{Object schemas}. Besides the $\type\obj$ pair, object schemas may additionally have the following: 

\smallskip
\noindent
- Pairs $\texttt{"minProperties":{i}}$ and $\texttt{"maxProperties":{j}}$, to specify that an object 
has to have at least $i$ and/or at most $j$ key-value pairs. 

\smallskip
\noindent
- a pair $\required[k_1,\dots,k_n]$, where each $k_i$ is a string value. This keyword 
mandates that the specified object values must contain pairs with keys $k_1,\dots,k_n$. 

\smallskip
\noindent
- a pair $\prop\{k_1:J_1,\dots,k_m:J_m\}$, where each $k_i$ a string value and 
each $J_i$ is itself a JSON Schema. This keyword states that the value of each pair with key $k_i$ must 
validate against schema $J_i$. 
%

\noindent
- A  pair $\patProp\{\texttt{"}e_1\texttt{":}J_1,\dots,\texttt{"}e_\ell\texttt{":}J_\ell\}$, where each $e_i$ is 
a regular expression over $\Sigma$ and each $J_i$ is a JSON schema. This keyword works just like \texttt{properties}, but 
now any value under any key that conforms to the expression $\exp_i$ must satisfy $J_i$. 

\smallskip
\noindent
- finally, the pair $\addProp J$, where $J$ is a JSON schema. This keyword 
presents a schema that must be satisfied by all values whose 
keys do not appear neither in \texttt{properties} nor conform to the language of an expression in \texttt{patternProperties}. 
For example, the following schema specifies objects 
where the value under "name" must be a string, the value under any key of the form \verb+a(b|c)a+ must be an even number, 
and the value under any key which is neither "name" nor conforms to the expression above must always be the number $1$. 

\beforeverbatim
{\footnotesize
\begin{verbatim}
{
   "type": "object", 
   "properties": {
      "name": {"type":"string"}, 
      }, 
   "patternProperties: {
      "a(b|c)a": {"type":"number", "multipleOf": 2}
      },     
   "additionalProperties: {
      "type": "number", 
      "minimum": 1
      "maximum": 1
   }
}
\end{verbatim}
}
\afterverbatim

\noindent 
\textbf{Array schemas}. Array schemas are specified with the \texttt{"type": "array"} keyword. For arrays there are two ways of specifying what 
kind of documents we find in arrays. We can use a pair $\items[J_1,\dots,J_n]$ to specify 
a document with an array of $n$ elements, where each $i$-th element must satisfy schema $J_i$. 
We can also use 
$\additionalItems J$ to specify that all elements in the array must satisfy schema $J$. If both keywords are used 
together, then we allow the array to have more values than those specified with items, 
as long as they agree with the schema specified in \texttt{additionalItems}. Finally, one can 
include the pair $\uniqueItems\true$ to force arrays whose elements are all distinct form each other.  
For example, the following schema validates against arrays of at least $2$ elements, where the first two 
are strings and the remaining ones, if they exists, are numbers. 

{\footnotesize
\begin{verbatim}
{
   "type": "array", 
   "items": [{"type":"string"}, {"type":"string"}], 
   "additionalItems": {"type":"number"},
   "uniqueItems":true
}
\end{verbatim}
}

\medskip
\noindent 
\textbf{Boolean combinations}. The last feature in JSON Schema are boolean combinations. These allow us
to specify that a document must validate against two schemas, against at least one schema, or 
that it must not validate against a schema. For example, the schema  
\verb+"not":{"type":"number","multipleOf":2}+ validates against any odd number, or any document which is not 
a number.

\subsection{JSON Schema Logic}

In order to capture the JSON Schema specification with a logical formalism, we isolate navigation and atomic tests into two different sets of operators. 
Let us start with atomic operations, which are basically a rewriting of most JSON Schema keywords 
into our framework. We allocate them in the set $\nodetests$.

Formally, $\nodetests$ contains the predicates $\isarray$, $\isobject$, $\isstring$, $\isnumber$ and $\uniq$, plus a predicate 
$\patternlogic(e)$ for each regular expression $e$ built from $\Sigma$, predicates $\minim(i)$ and $\maxim(i)$ for each integer $i$, 
a predicate $\mof(i)$ for each $i \geq 0$, predicates $\minch(k)$ and $\maxch(k)$ for each $k \geq 0$ and a predicate $\compare(A)$ for each JSON document $A$. 
The semantics of these predicates is given by the relation $\models$, that states whether an atomic predicate 
holds for a given node $n$ of a JSON $J$. 

\smallskip
\noindent
- $(J,n) \models \isarray $ iff. $n \in \Arr$. \ \ \ - $(J,n) \models \isobject $ iff. $n \in \Obj$. 

\smallskip
\noindent
- $(J,n) \models \isstring $ iff. $n \in \Str$. \ \ \ \ - $(J,n) \models \isnumber $ iff. $n \in \Int$. 

\smallskip
\noindent
- $(J,n) \models \patternlogic(e)$ iff. $\values(n)$ is a string in $L(e)$. 


\smallskip
\noindent
- $(J,n) \models \minim(i)$ iff. $\values(n)$ is a number greater than $i$.

\smallskip
\noindent
- $(J,n) \models \maxim(i)$ iff. $\values(n)$ is a number smaller than $i$. 

\smallskip
\noindent
- $(J,n) \models \mof(i)$ iff. $\values(n)$ is a multiple of $i$. 

\smallskip
\noindent
- $(J,n) \models \minch(i)$ iff. $n$ has at least $i$ children.

\smallskip
\noindent
- $(J,n) \models \maxch(i)$ iff. $n$ has at most $i$ children.


\smallskip
\noindent
- $(J,n) \models \uniq$ iff. $n \in \Arr$ and all of its children are different JSON values: for each $n'\neq n''$ such that $(n, p, n')$ and $(n, q, n'')$ belong to $\Achild$ 
we have that $\json(n') \neq \json(n'')$.

\smallskip
\noindent
- $(J,n) \models\ \compare(A) $ iff. $\json(n) = A$. 

\smallskip

With $\nodetests$ we cover all atomic features of JSON Schema. All that remains is the navigation, which in 
JSON Schema is given by the keywords properties, patternProperties, additionalProperties and required, 
for objects, and items and additionalItems for arrays. 
These forms of navigation are suggest using existential and universal modalities. For instance, 
the keyword \texttt{patternProperties} specifies a schema that must be validated by \emph{all} values whose keys 
satisfy a regular expression, and $\texttt{"required":}\ [w]$ demands that there must \emph{exist} 
a children with key $w$. 
Thus, to finish our logic, we augment our node tests with universal and existential modalities, as 
well as boolean combinations. 

\begin{definition}
Formulas in the \emph{JSON schema logic} (JSL) are expressions satisfying the grammar
\begin{equation*}
 %
 \begin{array}{lll}
\varphi,\psi & :=  &  \top \ |\ \neg \varphi\ \ |\ \varphi \wedge \psi \ |\ \phi\vee \psi \ |\  \psi \in \nodetests \ | \\ 
                  &  &       \bx_e \phi  \ |\  \bx_{i:j} \phi \ |\ \dm_{e} \phi \ |\ \dm_{i:j} \phi \\
                         
 \end{array}
\label{detjl-syntax}
\end{equation*}
where $e$ is a subset of $\Sigma^*$ (given as a regular expression),  $i \leq j$ are natural numbers, 
or $j=+\infty$ (signifying that we want any element of the array following $i$)and $A$ is an arbitrary JSON document. 
\end{definition}

As with JSON navigational logic, we can also obtain a deterministic version of JSL 
by restricting the syntax to use only modal operators  
$\bx_{w}$ and $\bx_{i}$, and $\dm_{w}$ and $\dm_{i}$; for 
a word $w \in \Sigma^*$ and a natural number $i$. 

\smallskip

The semantics is given by extending the relation $\models$.  


\smallskip
\noindent
- $(J,n) \models \top$ for every node $n$ in $J$.

\smallskip
\noindent
- $(J,n) \models \neg \phi$ iff. $(J,n) \nvDash \phi$.

\smallskip
\noindent
- $(J,n) \models \phi \wedge \psi$ iff. $(J,n) \models \phi$ and $(J,n) \models \psi$.

\smallskip
\noindent
- $(J,n) \models \phi \vee \psi$ iff. either $(J,n) \models \phi$ or $(J,n) \models \psi$.

\smallskip
\noindent
- $(J,n) \models \dm_{e}\ \phi$ iff. there is a word $w \in L(e)$ and a node $n'$ in $J$ such that 
$(n,w,n') \in \Ochild$ and $(J,n') \models \phi$

\smallskip
\noindent
- $(J,n) \models \dm_{i:j}\ \phi$ iff. there is  $i \leq p \leq j$ and a node $n'$ in $J$ such that 
$(n,p,n') \in \Achild$ and $(J,n') \models \phi$



\smallskip
\noindent
- $(J,n) \models \bx_{e}\ \phi$ iff. $(J,n') \models \phi$ for all nodes $n'$ such that 
$(n,w,n') \in \Ochild$ for some $w \in L(e)$.

\smallskip
\noindent
- $(J,n) \models \bx_{i:j}\ \phi$ iff. $(J,n') \models \phi$ for all nodes $n'$ such that 
$(n,p,n') \in \Achild$ for some $i \leq p \leq j$.


\smallskip

In order to present our results regarding JSL and JSON Schema, we abuse notation and write  
$J\models\psi$ whenever $(J, r)\models\psi$, where $r$ is the root of $J$. 

\medskip 
\noindent
\textbf{Expressive power}. As promised we show that JSL can capture JSON schema. 
In order to present this result we informally speak of the validation relation of JSON Schema, and 
say that a JSON $S$ validates against $J$ whenever $J$ is in accordance to all keywords present in $S$. 
We refer to \cite{PRSUV16} for more details on the semantics. 
As usual, we say that JSON Schema and JSL are equivalent in expressive power if 
for any JSON Schema $S$ there exists a $JSL$ formula $\psi_S\in \mathcal L$ such that for every JSON document $J$ we have that $J$ validates against $S$ if and only if $J\models\psi_S$; and conversely, for any expression $\phi \in \mathcal L$ there exists a JSON Schema $S_\phi$  such that for every JSON document $J$ we have that $J\models\psi$ if and only if $J$ validates against $S$. 

\begin{theorem}
  \label{schema2logic}
JSL and JSON Schema are equivalent in expressive power.
\end{theorem} 



\smallskip 
\noindent
\textbf{Comparing JSL and JNL}. Next, we consider how the navigation logic of Section \ref{sec:nav} compares to the schema logic JSL.
Even though the starting point of the two logics is different, we next show that
the two logics are essentially the same, their expressivity differing simply because of the different atomic predicates. 
More precisely, we have:

\begin{theorem}\label{thm-jsn_jnl}
Non-deterministic JNL not using the equality $\EQ(\alpha,\beta)$ and non-deterministic JSL using only the node test $\sim(A)$ are equivalent in terms of  expressive power. More precisely:
\begin{itemize}
\item For every formula $\varphi^S$ in JSL there exists a unary formula $\varphi^N$ in JNL such that for every JSON $J$: $$\sem{\varphi^N}_J = \{n\in J \mid (J,n)\models \varphi^S\}.$$
\item For every unary  formula $\varphi^N$ in JNL there exists a $\varphi^S$ in JSL such that for every JSON $J$: $$\sem{\varphi^N}_J = \{n\in J \mid (J,n)\models \varphi^S\}.$$
\end{itemize}
\end{theorem}

In the proof above, we also show that going from JSL to JNL takes only polynomial time, while the transition in the other direction can be exponential. This implies that the upper bounds for JNL are valid for JSL, while the lower bounds transfer in the opposite direction (without taking into account node tests). 

\medskip 
\noindent
\textbf{Algorithmic Properties}. Since JSL is designed to be a schema logic to validate trees, we 
specify a boolean {\sc Evaluation} problem: the input is a JSON $J$ and a JSL expression $\phi$, and 
we decide whether $J \models \phi$.


\begin{proposition}
\label{prop-jsl-eval}
The \textsc{Evaluation} problem for JSL can be solved in time $O(|J|^2 \cdot |\phi|)$, and in 
$O(|J| \cdot |\phi|)$ when $\phi$ does not use the uniqueItems keyword. 
\end{proposition}

From Theorem \ref{schema2logic} we obtain as a corollary that the validation problem for JSON Schema 
has the same bounds. This was  already shown in \cite{PRSUV16}.
Next, we study the \textsc{Satisfiability} problem, which receives a formula $\phi$ as input and 
consists of deciding  whether there is any JSON tree $J$ such that $J \models \phi$. 
Here we need to be careful with the encoding we choose, as the interplay between 
$\uniq$ and $\dm_i \top$ immediately forces a node with exponentially many different children when 
$i$ is given in binary. In terms of results, this means that our algorithms raise by one exponential, 
although we don't know if this increase is actually unavoidable. 


\begin{proposition}
\label{prop-jsl-sat}
The \textsc{Satisfiability} problem for JSL is in \expspace, and \pspace-complete for expressions without $\uniq$. 
\end{proposition}



We remark that the \textsc{Satisfiability} problem is important in the context of JSON Schema. For example, 
the community has repeatedly stated the need for algorithms that can learn JSON Schemas from examples. We believe that understanding 
basic tasks such as satisfiability are the first steps to proceed in this direction. 

\subsection{Adding recursion}

The JSON Schema specification also allows defining statements that 
will be reused later on. We have so far ignored this functionality, and 
to capture it we will need to define the same operator in our logic. 
As we will see, this lifts the expressive power of JSL away from even the recursive version 
of our navigational logic; and is very similar to certain forms of tree automata. 

Let us explain first how recursion is added into JSON Schema.  The idea is to allow to an additional 
keyword, of the form \verb+{$ref: <path>}+, where \verb+<path>+ is a navigation instruction. 
This instruction is used within the same document to fetch other schemas that have been predefined 
in a reserved \verb+definitions+ section\footnote{Definitions and references can also be used to fetch schemas in 
different documents or even domains; here we just focus on the recursive functionality.}.
For example, the following schema validates against any JSON which is not a string 
following the specified pattern. 

\beforeverbatim
{\footnotesize
\begin{verbatim}
{
    "definitions": {
        "email": {
            "type": "string", 
            "pattern": "[A-z]*@ciws.cl"
        }
    }, 
    "not": {"$ref": "#/definitions/email"}
}
\end{verbatim}
}
\afterverbatim

As we have mentioned, different schemas are defined under the \verb+definitions+ section of the JSON document, and 
these definitions can be reused using the \verb+$ref+ keyword. In the example above, we use 
\verb+{"$ref": "#/definitions/email"}+ to retrieve the schema 
\verb+{"type": "string",+ \verb+"pattern": "[A-z]*@ciws.cl"}+. These definitions can be nested within 
each other, but semantics is currently defined only for a fragment with limited recursion. We come back to this issue 
after we define a logic capturing these schemas. 

\smallskip
\noindent
\textbf{Recursive JSL}. The idea of this logic is to capture the recursive functionalities present in JSON Schema: 
there is a special section where one can define new formulas, which can be then re-used in other formulas.  

Fix an infinite set $\Gamma = \{ \gamma_1,\gamma_2,\gamma_3,\dots\}$ of symbols. 
A \emph{recursive} JSL formula is an expression of the form 
\begin{eqnarray}
\nonumber
\gamma_1 &= &\phi_1 \\
\nonumber
\gamma_2 &= &\phi_2 \\
\nonumber
&\vdots &\\
\nonumber
\gamma_m &=& \phi_m \\
\psi & &  \label{eqn-rec-jsl}
\end{eqnarray}
where each $\gamma_i$ is a symbol in $\Gamma$ and $\phi_i$, $\psi$ are JSL formulas over the extended syntax that allows 
$\gamma_1,\dots,\gamma_m$ as atomic predicates. Here each equality $\gamma_i = \phi_i$ is called a definition, and $\psi$ is called the base expression. 

The intuition is that each $\gamma_i$ is one of the references of JSON Schema. 
Before moving to the semantics, let us show the way recursive JSL formulas work.

\begin{example}
\label{exa-paths-even-length}
Consider the expression $\Delta$ given by 
$$\begin{array}{l}
\gamma_1 = \bx_{\Sigma^*} \gamma_2\\
\gamma_2 = (\dm_{\Sigma^*} \top) \wedge (\bx_{\Sigma^*} \gamma_1 ) \\
\gamma_1
\end{array}$$
Intuitively, $\gamma_1$ holds in a node $n$ if $n$ either has no children, or if $\gamma_2$ holds all of its its children. 
On the other hand, $\gamma_2$ holds  in a node $n$ if this node has at least one child, and $\gamma_1$ holds in all children of $n$. 
Finally, the  base expression simply states that $\gamma_1$ has to hold in the root of the document (recall that a schema statement is evaluated at the root). The intuition for $\Delta$ is, then, that it should hold on every tree such that each path from the root to the 
leaves is of even length. 
\end{example}

\smallskip
\noindent
\textbf{Well-formed recursive JSL}. As usual in formalisms that mix recursion and negation, giving a formal semantics for 
JSON Schema is not a straightforward task. 
As an example of the problems we face, consider the following  JSL expression. 
$$\begin{array}{l}
\gamma_1 = \neg  \gamma_1\\
\gamma_1
\end{array}$$
Of course, we can always give a logical interpretation to this formula, but 
we argue that this expression does not specify any real restriction on JSON documents, and thus 
any semantics we establish will not be intuitive from the point of view of defining schemas for JSON. 

The most straightforward way to avoid these issues is by imposing a strict acyclicity condition on definitions. 
However, we can in fact work with a much milder restriction that we call \emph{well-formedness}. 
This restriction was introduced in \cite{PRSUV16} for the case of JSON Schema, but  
we can seemingly define well-formed recursive JSL expressions, which can then be shown to capture well-formed recursive JSON Schemas. 

For a recursive JSL expression $\Delta$ of the form (\ref{eqn-rec-jsl}) defined above, we define the precedence graph of $\Delta$ as a graph whose nodes 
are $\gamma_1,\dots,\gamma_m$ and where there is an edge from $\gamma_i$ to $\gamma_j$ if $\gamma_j$ appears in the expression $\phi_i$ of 
the definition  $\gamma_i = \phi_i$, but only if this appearance is {\em not under the scope of a modal operator}. 
We then say that $\Delta$ is \emph{well-formed} if its precedence graph is acyclic. 

\begin{example}
The definition $\gamma_1 = \neg  \gamma_1$ clearly creates a cyclic precedence graph, as the construction mandates a self-loop in the node 
corresponding to $\gamma_1$. On the other hand, the recursive JSL expression in Example \ref{exa-paths-even-length} is indeed 
well-formed. It does introduce cycles in the definitions, but no edges are added into the precedence graph of such expression 
because formula symbols are always under the scope of a modal operator. 
\end{example}

\smallskip
\noindent
\textbf{Semantics}. How do we then define the semantics of well-formed expressions? 
If $\Delta$ is a recursive JSL expression that is completely acyclic, then we can simply replace the symbols in 
$\Gamma$ by their respective definitions. But we cannot do this with every well-formed expression, because 
some of them can have cycles in the definitions (under a scope of a modal operator), and thus we would never stop replacing symbols. This is the case, for instance, with the expression in Example \ref{exa-paths-even-length}. 

However, the key thing to notice is that we only need to do this as many times as the height of the JSON tree 
we are trying to validate. More precisely, let $J$ be a JSON tree 
of height $h$ and $\Delta$ a well-formed recursive JSL expression of the form (\ref{eqn-rec-jsl}). 
Construct a (non-recursive) JSL expression $\unfold_J(\psi)$ by replacing each symbol $\gamma_i$ in $\psi$ by the corresponding definition $\phi_i$, 
but stop once every symbol from $\Gamma$ is under the scope of at least $h+1$ modal operators. Afterwards, replace all the remaining symbols $\gamma_i$ for the symbol $\bot$ (shorthand for $\neg \top$). 
If $\Delta$ is well-formed, then this procedure is guaranteed to stop, because every time we come back to the same symbol $\gamma$ 
when replacing, we know that this symbol has to be under the scope of at least one more modal operator. 
Then, we define the satisfaction relation based on the satisfaction of the constructed formula $\psi$, so that $J \models \Delta$ 
if and only if $J \models \unfold_J(\psi)$. 

\begin{example}
Suppose that we need to evaluate the expression in Example \ref{exa-paths-even-length} over the tree $J$ of height 4. 
Then we would only need to keep rewriting $\psi$ until all of $\gamma_1$ and $\gamma_2$ are under at least 5 modal operators. 
The query obtained corresponds to 
$$\bx_{\Sigma*} \big(\dm_{\Sigma^*} \top \wedge \bx_{\Sigma*}\bx_{\Sigma*} (\dm_{\Sigma^*} \top \wedge \bx_{\Sigma*}\bx_{\Sigma*} \gamma_2) \big)$$
To finalise we create the formula $\unfold_J(\gamma_1)$ by replacing the symbol $\gamma_2$ in the expression above for $\bot$. Now the evaluation of the original expression over $J$ corresponds to the evaluation 
of $\unfold_J(\gamma_1)$ over $J$. 
\end{example}

\smallskip
\noindent
\textbf{Expressive Power}. As promised, one can show that recursive JSL captures the recursive definition of JSON Schema. 
We use the same notation as for Theorem  \ref{schema2logic}.


\begin{theorem}
  \label{schema2logic-recursive}
 Well-formed recursive JSL and well-formed recursive JSON Schema are equivalent in expressive power. 
\end{theorem} 

Comparing JSL to navigational logic JNL, we can again show that different atomic predicates force them to have different expressive power: for instance JSL can not express the $\EQ(\alpha,\beta)$ operator, while JNL can not cover unique items. On the other hand, if $\EQ(\alpha,\beta)$ is not allowed, then we conjecture recursive non-deterministic JNL to be strictly less expressive than recursive JSL without unique items, due to the more powerful form of recursion available in the language.

\smallskip

%

We finish with a few remarks on the expressive power of recursive JSL expressions. First, let us note that, even if one would 
like to compare recursive JSL or JSON Schema 
against XML Schema definitions such as DTDs or XML Schema, or even to Monadic Second Order (MSO), 
we cannot do it in a direct way since the models 
of XML and JSON have several important differences, and it is not immediate to express JSON as a relational 
structure. 
However, just for the sake of establishing a comparison we can assume that the set $\Sigma^*$ of possible keys is 
fixed and finite, and that we do not consider arrays. We can then just focus on a relational representation of 
JSON that has one binary relation $\Ochild_w$ for each word $w$ in our set of keys, plus all predicates specified in $\nodetests$. 
We can then show the following (as usual MSO is said be equivalent to JSL if for every JSL expression we 
can create a boolean MSO formula accepting the same trees, and vice-versa): 
\begin{proposition}
\label{prop-jsl-mso}
Well-formed recursive JSL is equivalent to MSO, if the set $\Sigma^*$ of possible keys is 
fixed and finite, and that we do not consider arrays.
\end{proposition}

What about arrays? The first observation is that the presence of arrays introduces the atomic test $\uniq$, which 
can express properties not definable in MSO: 
\begin{example}
The following recursive JSL expression accepts only JSON documents representing complete binary trees:  
$$\begin{array}{l}
\gamma = \neg (\dm_1 \top) \vee \big(\minch(2) \wedge \maxch(2) \wedge \\
\hspace{4cm} \neg \uniq \wedge \bx_{1:2}\gamma \big) \\ 
\gamma
\end{array}$$
Every node is an array with either no child or with two children both satisfying $\gamma$, 
and the $\neg \uniq$ restriction forces the two children of each node to be equal.
\end{example}

Even if we rule out $\uniq$, it is still not easy to exactly pinpoint what JSL can do. On one hand, JSON arrays are ordered, in the 
sense that the first item can be distinguished from the second. But on the other hand the reasoning between elements in arrays in JSON is 
very limited. For example, we cannot use JSL to specify that after an element satisfying a formula $\phi$ we must have an element satisfying another 
formula $\psi$. 


\smallskip
\noindent
\textbf{Evaluation}. 
The first thing to note is that the semantics for recursive JSL, as defined previously, 
leads to very inefficient evaluation algorithms: the rewriting $\unfold_J(\psi)$
may well be of exponential size with respect to the original query, even if $J$ contains a single node. 
However, we can show that evaluating recursive JSL expressions remains in \ptime\ in combined complexity, although 
the succinctness introduced by the possibility of reusing definitions makes the problem \ptime-hard. 
Our algorithm consists on evaluating all subtrees of $J$ in a bottom-up fashion, proceeding to higher  
height levels of $J$ only when all the previous levels have already been computed. The algorithm resembles 
the evaluation of Datalog programs with stratified negation. 

\begin{proposition}
\label{prop-eval-recursive-jsl}
The \textsc{Evaluation} problem for recursive JSL expressions over JSON trees is \ptime-complete in combined complexity. 
\end{proposition}

\smallskip
\noindent
\textbf{Satisfiability}. 
The most common way of building a satisfiability algorithm in schema formalisms for trees is to show that 
they are equivalent to some class of tree automata whose non-emptiness problem can be shown 
to be decidable. We use the same ideas, albeit we need to introduce 
a specific model of automata that can capture our formalism. Interestingly, we show that the $\uniq$ predicate can also 
be handled in this case, albeit with an exponential blowup. To show this 
we encode $\uniq$ as a special local constraint, as done in e.g. \cite{BT92,KL07}. Again, we do not know 
if this blowup is unavoidable. 



\begin{proposition}
\label{prop-sat-recursive-jsl}
The \textsc{Satisfiability} problem for recursive JSL expressions is in \twoexptime, and \exptime-complete 
for expressions without the $\uniq$ predicate. 
\end{proposition} 

\section{Future perspectives}\label{sec:future}

In this work we present a first attempt to formally study the JSON data format. To this end, we describe the underlying data model for JSON, and introduce logical formalisms which capture the way JSON data is accessed and controlled in practice. Through our results we emphasise how the new features present in JSON,affect classical results known from the XML context, and highlight  that there is a need for developing new techniques to tackle main static tasks for JSON languages. And while some of these features have been consider in the past (e.g. comparing subtrees \cite{BogaertT92, KariantoL07}, or an infinite number of keys \cite{BoiretHNT14}), it still not entirely clear how these properties mix with the deterministic structure of JSON trees, thus providing an interesting ground for future work.

Apart from these fundamental problems that need to be tackled, there is also a series of practical aspects of JSON that we did not consider. In particular, we identify three areas that we believe warrant further study, and where the formal framework we propose could be of use in understanding the underlying problems.


\smallskip
\noindent
\textbf{MongoDB's projection}. While we have presented a navigational logic that can capture 
the way MongoDB filters JSON documents within its find function, we have left out the second argument of the find function, known as projection. In its essence, the idea of the projection argument is to select only those subtrees of input documents that can be reached by certain navigation instructions, thus defining a JSON to JSON transformation. Although well-defined, the projection in MongoDB is quite limited in expressive power, and does not allow a lot of interaction between filtering and projecting. We believe that this is an interesting ground for future work, as there are many fundamental questions regarding the expressive power of these transformations, and their possible interactions with schema definitions.

\smallskip
\noindent
\textbf{Streaming}. Another important line of future work is streaming. Indeed, the widespread use of JSON as 
a means of communicating information through the Web demands the usage of streaming techniques to query JSON 
documents or validate document against schemas. Streaming applications most surely will be related to APIs, either 
to be able to query data fetched from an API without resorting to store the data (for example if we are in a mobile environment) or 
to validate JSONs on-the-fly before they are received by APIs. In contrast with XML (see e.g. \cite{SS07}, we suspect that our deterministic versions 
of both JNL and JSL might actually be shown to be evaluated in a streaming context with constant memory requirements when tree equality is excluded from the language. 

\smallskip
\noindent
\textbf{Documenting APIs}. One of the main uses of JSON Schema is the \emph{Open API initiative} \cite{openapi}, an endeavour founded in early 2016 that intends to build an open documentation of RESTful APIs available worldwide. This initiative uses JSON Schema to specify the inputs and the outputs of a large number of APIs, and is trying to describe which parts of the API output actually come from the input.
A formal perspective in these tasks will certainly be well appreciated by the community. Furthermore, there is also the problem of documenting how JSON APIs interact with real databases. For example, when large organisations expose data they sometimes use APIs instead of proper database views, even if it is for internal uses. This immediately raises the  problem of exchanging data that may now not reside on real databases, but is exposed only by means of JSON. We believe these  problems raise interesting questions for future work.

\bibliographystyle{plain}
\bibliography{json}

\begin{thebibliography}{10}

\bibitem{AI00}
N.~Alechina and N.~Immerman.
\newblock Reachability logic: An efficient fragment of transitive closure
  logic.
\newblock {\em Logic Journal of the IGPL}, 8(3):325--337, 2000.

\bibitem{Arango}
{AranbgoDB Inc.}
\newblock {The ArangoDB database}.
\newblock \url{https://www.arangodb.com/}, 2016.

\bibitem{BPR12}
Pablo Barcel{\'{o}}, Jorge P{\'{e}}rez, and Juan~L. Reutter.
\newblock Relative expressiveness of nested regular expressions.
\newblock In {\em Proceedings of the 6th Alberto Mendelzon International
  Workshop on Foundations of Data Management, Ouro Preto, Brazil, June 27-30,
  2012}, pages 180--195, 2012.

\bibitem{BT92}
Bruno Bogaert and Sophie Tison.
\newblock Equality and disequality constraints on direct subterms in tree
  automata.
\newblock In {\em Annual Symposium on Theoretical Aspects of Computer Science},
  pages 159--171. Springer, 1992.

\bibitem{BogaertT92}
Bruno Bogaert and Sophie Tison.
\newblock Equality and disequality constraints on direct subterms in tree
  automata.
\newblock In {\em {STACS} 92, 9th Annual Symposium on Theoretical Aspects of
  Computer Science, Cachan, France, February 13-15, 1992, Proceedings}, pages
  161--171, 1992.

\bibitem{BoiretHNT14}
Adrien Boiret, Vincent Hugot, Joachim Niehren, and Ralf Treinen.
\newblock Deterministic automata for unordered trees.
\newblock In {\em Proceedings Fifth International Symposium on Games, Automata,
  Logics and Formal Verification, GandALF 2014, Verona, Italy, September 10-12,
  2014.}, pages 189--202, 2014.

\bibitem{bojanczyk2009two}
Mikoaj Boja{\'n}czyk, Anca Muscholl, Thomas Schwentick, and Luc Segoufin.
\newblock Two-variable logic on data trees and xml reasoning.
\newblock {\em Journal of the ACM (JACM)}, 56(3):13, 2009.

\bibitem{BCCRX16}
Elena Botoeva, Diego Calvanese, Benjamin Cogrel, Martin Rezk, and Guohui Xiao.
\newblock A formal presentation of mongodb (extended version).
\newblock {\em CoRR}, abs/1603.09291, 2016.

\bibitem{json}
Tim Bray.
\newblock {The JavaScript Object Notation (JSON) Data Interchange Format}.
\newblock 2014.

\bibitem{CleavelandS93}
R.~Cleaveland and B.~Steffen.
\newblock A linear-time model-checking algorithm for the alternation-free modal
  mu-calculus.
\newblock {\em Formal Methods in System Design}, 2(2):121--147, 1993.

\bibitem{tata2007}
H.~Comon, M.~Dauchet, R.~Gilleron, C.~L\"oding, F.~Jacquemard, D.~Lugiez,
  S.~Tison, and M.~Tommasi.
\newblock Tree automata techniques and applications, 2007.
\newblock release October, 12th 2007.

\bibitem{json-ecma}
{ECMA}.
\newblock {The JSON Data Interchange Format}.
\newblock
  \url{http://www.ecma-international.org/publications/standards/Ecma-404.htm},
  2013.

\bibitem{Figueira10}
Diego Figueira.
\newblock {\em Reasoning on words and trees with data. (Raisonnement sur mots
  et arbres avec donn{\'{e}}es)}.
\newblock PhD thesis, {\'{E}}cole normale sup{\'{e}}rieure de Cachan, France,
  2010.

\bibitem{python}
Python~Software Foundation.
\newblock {Python programming language}.
\newblock \url{https://www.python.org/}, 2016.

\bibitem{jpath}
Stefan G\"ossner and Stephen Frank.
\newblock {JSONPath}.
\newblock \url{http://goessner.net/articles/JsonPath/}, 2007.

\bibitem{GottlobKP05}
Georg Gottlob, Christoph Koch, and Reinhard Pichler.
\newblock Efficient algorithms for processing xpath queries.
\newblock {\em {ACM} Trans. Database Syst.}, 30(2):444--491, 2005.

\bibitem{Gottlob2006}
Georg Gottlob, Christoph Koch, and Klaus~U. Schulz.
\newblock Conjunctive queries over trees.
\newblock {\em J. ACM}, 53(2):238--272, March 2006.

\bibitem{json-ietf}
{Internet Engineering Task Force (IETF)}.
\newblock {The JavaScript Object Notation (JSON) Data Interchange Format}.
\newblock \url{https://tools.ietf.org/html/rfc7159}, March 2014.

\bibitem{jsonschema}
json-schema.org: The home of json schema.
\newblock \url{http://json-schema.org/}.

\bibitem{KL07}
Wong Karianto and Christof L{\"o}ding.
\newblock Unranked tree automata with sibling equalities and disequalities.
\newblock In {\em International Colloquium on Automata, Languages, and
  Programming}, pages 875--887. Springer, 2007.

\bibitem{LMV16}
Leonid Libkin, Wim Martens, and Domagoj Vrgo\v{c}.
\newblock Querying graphs with data.
\newblock {\em J. {ACM}}, 63(2):14, 2016.

\bibitem{LiuHM14}
Zhen~Hua Liu, Beda~Christoph Hammerschmidt, and Doug McMahon.
\newblock {JSON} data management: supporting schema-less development in
  {RDBMS}.
\newblock In {\em International Conference on Management of Data, {SIGMOD}
  2014, Snowbird, UT, USA, June 22-27, 2014}, pages 1247--1258, 2014.

\bibitem{mongoDB}
{MongoDB Inc.}
\newblock {The MongoDB3.0 Manual}.
\newblock \url{https://docs.mongodb.org/manual/}, 2015.

\bibitem{nurseitov2009comparison}
Nurzhan Nurseitov, Michael Paulson, Randall Reynolds, and Clemente Izurieta.
\newblock Comparison of json and xml data interchange formats: A case study.
\newblock {\em Caine}, 2009:157--162, 2009.

\bibitem{sqlpp}
Kian~Win Ong, Yannis Papakonstantinou, and Romain Vernoux.
\newblock The {SQL++} semi-structured data model and query language: {A}
  capabilities survey of sql-on-hadoop, nosql and newsql databases.
\newblock {\em CoRR}, abs/1405.3631, 2014.

\bibitem{openapi}
{The Open API Initiative}.
\newblock \url{https://openapis.org/}, 2016.

\bibitem{Orient}
{OrientDB LTD}.
\newblock {The OrientDB database}.
\newblock \url{http://orientdb.com/}, 2016.

\bibitem{Parys09}
Pawel Parys.
\newblock Xpath evaluation in linear time with polynomial combined complexity.
\newblock In {\em Proceedings of the Twenty-Eigth {ACM} {SIGMOD-SIGACT-SIGART}
  Symposium on Principles of Database Systems, {PODS} 2009, June 19 - July 1,
  2009, Providence, Rhode Island, {USA}}, pages 55--64, 2009.

\bibitem{PRSUV16}
Felipe Pezoa, Juan~L. Reutter, Fernando Suarez, Mart{\'{\i}}n Ugarte, and
  Domagoj Vrgo\v{c}.
\newblock Foundations of {JSON} schema.
\newblock In {\em Proceedings of the 25th International Conference on World
  Wide Web, {WWW} 2016, Montreal, Canada, April 11 - 15, 2016}, pages 263--273,
  2016.

\bibitem{jsoniq}
Jonathan Robie, Ghislain Fourny, Matthias Brantner, Daniela Florescu, Till
  Westmann, and Markos Zaharioudakis.
\newblock {JSONiq: The JSON Query Language}.
\newblock \url{http://www.jsoniq.org/}, 2016.

\bibitem{SS07}
Luc Segoufin and Cristina Sirangelo.
\newblock Constant-memory validation of streaming xml documents against dtds.
\newblock In {\em International Conference on Database Theory}, pages 299--313.
  Springer, 2007.

\bibitem{CouchDB}
{The Apache Software Foundation}.
\newblock {Apache CouchDB}.
\newblock \url{http://couchdb.apache.org/}, 2015.

\bibitem{neo4j}
{The Neo4j Team}.
\newblock {The Neo4j Manual v3.0}.
\newblock \url{http://neo4j.com/docs/stable/}, 2016.

\bibitem{KariantoL07}
Karianto Wong and Christof L{\"{o}}ding.
\newblock Unranked tree automata with sibling equalities and disequalities.
\newblock In {\em Automata, Languages and Programming, 34th International
  Colloquium, {ICALP} 2007, Wroclaw, Poland, July 9-13, 2007, Proceedings},
  pages 875--887, 2007.

\end{thebibliography}


\newpage 
\onecolumn
\appendix

\section{Proofs and Additional Results}

\smallskip
\noindent
\textbf{Proposition} \ref{prop-JNL-eval}.
The \textsc{Evaluation} problem for JNL 
can be solved in time $O(|J|\cdot|\phi|)$.

\begin{proof}
We demonstrate that for a Json tree $J$ and a unary formula $\phi$, the evaluation of $\phi$ over $\texttt{J}$ is in $O(|\phi| * |J|)$.
For that we first translate the unary formula into a non recursive monadic datalog program with stratified negation in the style of \cite{Gottlob2006} for XML trees. 
For that, we need to define a relationnal structure on which the JSON tree is encoded. For each key $k$, there is a binary relation $k(x,y)$ such that $y$ is accessible from $x$ by the key $k$. Let $m$ be the maximal size of the arrays in the trees, then for each integer $i$ between $0$ to $m-1$, there is a binary relation $i$.
There is a binary relation $Self(x,y)$ such that $x$ is the same object of $y$ and accessible by the same key from the same parent. The binary relations presented until now form the set of navigational relations. The 
Moreover, there is one binary relation $Eq$ comparing the values of objects. Finally, for each value $v$, there is a unary relation  $v(x)$ such that the value of $x$  is equal to $v$.
By induction on the structure of the formula, we can prove that we can construct that a monadic datalog program with stratified negation only over intentional predicates.  
Moreover, for each rule of the program, the body of this rule is a tree-shape connected query by using the navigational relations $k$ and $i$ where $k$ and $i$ are keys or integers.  We denote this program $P_\phi$.

A unary conjunctive query is a \emph{tree query} iff 
\begin{itemize}
\item the query is connected regarding the navigational relations
\item the query restricted of the navigational relations describes a tree such that the root of this tree is the free variable.
\end{itemize}

A Datalog program is a \emph{JSON program} iff 
\begin{itemize}
\item It is monadic
\item The graph of dependence of the intentional relation is a tree
\item The intentional relation appearing in the head of a rule is monadic at the exception of the goal predicate which is a $0$ arity relation. 
\item The goal predicate appears only on the head of one rule and the body of this rule is a query of the form $Root(x) \wedge \phi(x)$ where $\phi(x)$ is a tree query.
\item For each rule, the body of this rule is a tree query. 
\item The negation of the program is stratified and the only negated atoms are unary intentional predicates. 
\end{itemize}

We need now to prove the two  key lemmas in order to efficiently compute the evalution $\phi$ over $J$.

\begin{lemma}
\label{lem:treequery}
Let $Q$ be a tree query over the JSON signature. Let $J$ be a tree. Let $o$ be an object of this tree.
There exists a unique valuation demonstrating that $o$ is answer of $Q$ applied to $P$. This valuation can be computed in $O(|\phi| \cdot |J|)$
\end{lemma}

The idea of this proof is to build inductively on the tree structure of the query restricted to the navigational relations. It starts from the object $o$ that has to map the root of this tree.
The propagation case relies heavly on the fact that the first attribute of each navigational relation is a primary key. Because the query is connected regarding the navigational relations then a such valuation over the navigationa relations of the query defines a valuation over all the variables of the query. 
Finally, to obtain the good complexity, we cannot compute the relation $Eq$ for all the nodes of the $J$ and therefore the comparaison of subtrees are computed online.
We now consider $\nu(Q)$ the set of ground facts obtained from the valuation $\nu$ obtained on the evaluation over the ground facts. For each ground atom $Eq(o,o')$, we compare the subtrees which can be done in linear time.


\begin{lemma}
Let $P$ be a JSON program. Let $J$ be a JSON tree. 
Let $\mathcal{T}(P,J)$ be the set of proofs trees of the acceptance of $P$ by $J$. Let $p(n)$ and $p(n')$ be two facts appearing in different proof trees of $\mathcal{T}(P,J)$. Then $n$ is $n'$ are exactly the same objects reachable from the same parent with the same key.
\end{lemma}

This lemma is proved by induction on dependence tree of the intentional relations. It relies mainly on Lemma \ref{lem:treequery}.

We now presente the algorithm evaluation. This algorithm is really differrent than the classical evaluation of Datalog over an instance.
The algorithm is built inductively on the structure of the  graph of dependence of the program and the tree and the goal is to ground the rules.
It relies heavily on Lemma \ref{lem:treequery}. Once, the rules are grounded if possible, we just need to check that the goal can be deduced. As the grounded program is of size $O(|\phi|\cdot|J|)$ 
, then checking that goal can be deduced in linear time in the size of the program.
This finishes the proof.

\end{proof}

\smallskip
\noindent
\textbf{Proposition} \ref{prop-JNL-sat}
The \textsc{Satisfiability} problem for JNL is \np-complete. It is \np-hard even for formulas not using negation nor the equality operator.  

\begin{proof}
For the upper bound one just guesses a polynomial time witness. 
For JNL the height of the models is bounded by the size of the formula. However, the width may be of exponential size 
because $\child_i$ demands a node with exponentially many nodes if $i$ is in binary. 
However, before guessing a witness we preprocess 
the formulas so that we replace, for each level, all binary numbers for the unary numbers $1,2,3,$ etc. To do this we just need to keep into account 
the ordering of the original numbers, and then respect this ordering with our replacement. For example, for the first level (that is, all numbers appearing in a $\child_i$ 
operator not under the scope of another $\child$ operator) we just order all numbers in $\child_i$ operators, and then respecting this ordering we replace the lowest 
such number by $1$, the second by $2$, and so on. This operation is polynomial in the formula, and after we finish with this preprocessing 
we get a formula which always has exponential witnesses. Having guesses the witness we evaluate our formula polynomial time as shown in Proposition \ref{prop-JNL-eval}.


\textbf{Hardness}. Reduction is from 3SAT. Let $\phi$ be a propositional formula in 3CNF using variables 
$p_1,\dots,p_n$ and clauses $C_1,\dots,C_m$. 
For each $p_i$ we define the formula $\theta_{p_1} = (\child_{p_1} [\child_1]) \vee (\child_{p_1} [\child_w])$, with $w$ a fresh string, 
with the intention of allowing, as models, all valuations of each of the $p_i$'s: if the object under key $p_1$ is an array, then 
we will interpret this as $p_i$ being assigned the value true, and if it is an object we will interpret that 
$p_i$ was assigned the value false.  
Moreover, for each 
of the clauses $C_j$ that uses variables $a$, $b$ and $c$, define $\gamma_{C_J}$ as 
$\child_{a}S_a \vee \child_{b}S_b \vee \child_{c}S_c$, where 
each $S_a$, $S_b$ and $S_c$ is either $[\child_1]$, if $a$ (respectively, $b$ or $c$) appears positively in $C_j$, and 
$[\child_w]$ otherwise. 

Recall that all edges leaving from array nodes are labelled with natural numbers, and all edges leaving from object nodes 
are labelled with strings. This means that for any json $J$ it must be the case that $\sem{[X_1]}_J$ and $\sem{[X_a]}_J$ 
are always disjoint. Moreover,  recall that an object cannot have two pairs with identical keys, 
so a node in a JSON trees cannot have two children under an edge with the same label. 
With these two remarks it is then immediate to see that $\phi$ is satisfiable if and only if the following JNL expression is satisfiable:
$$\bigwedge_{1 \leq i \leq n} \theta_{p_i} \wedge \bigwedge_{1 \leq \ell \leq m} \gamma_{C_\ell}$$
\end{proof}

\smallskip
\noindent
\textbf{Proposition} \ref{prop-JNL-rec-eval}
The evaluation problem for JNL with non-determinism and recursion can be solved in cubic time, and in linear time if 
the formula does not use predicate $\EQ(\alpha,\beta)$.

\begin{proof}
When our formula does not use the $\EQ(\alpha,\beta)$ operator, we can reuse the classical model checking algorithm from PDL \cite{AI00,CleavelandS93}, since our logic is a syntactic variant of PDL and JSON trees can be viewed as generalisations of Kripke structures. Some small changes are needed though in order to accommodate the specifics of JSON and of our syntax. First, arrays can be treated as usual nodes, with the edges accessing them being labelled by numbers. Second, for the formula $X_e$, where $e$ is a regular expression, JNL traverses a single edge (i.e. the regular expression is not applied as the Kleene star over the formulas). To accommodate this, we can mark each edge of our tree with an expression $e$ such that the label of the edge belongs to the language of $e$.  Since checking membership of a label $l$ in the language of the expression $e$ can be done in $O(|e|\cdot |l|)$, and the sum of the length of all the labels is less than the size of the model (as we have to at least write them all down), this can be done in $O(|e|\cdot |J|)$. We now repeat this for every regular expression appearing in our JNL formula. Since the number of expressions is linear in the size of the formula the preprocessing takes linear time. 
We can now run the classical PDL model checking over this extended structure treating regular expressions as ordinary edge labels. 
In the case that $\EQ(\alpha,\beta)$ is present, we can utilize the PDL-algorithm of \cite{LMV16}, which runs in cubic time.
\end{proof}

\smallskip
\noindent
\textbf{Proposition} \ref{prop-JNL-rec-sat-und}
The \textsc{Satisfiability} problem is undecidable for non-deterministic recursive JNL formulas, even if they do not use negation.

\begin{proof}
We prove that the satisfiability for positive JSL queries is undecidable.
We reduce the problem of emptiness of two counters machine.
Let $M$ be a two counters machine. The initial state is $q_0$. The final state is $q_f$.
We encode a configuration of the machine by an object $o$ having two keys $c_1$ and $c_2$ representing the counter $C_1$ and $C_2$. The value of the counte $C_i$  is the height of the subtree reachable from the key $c_i$. The  empty counter is represented by the string "0".
Moreover, $o$ has two other keys : $state$ given the string representing the state of the configuration and the key $next$ given the object representing the next configuration.
The query $Q$ is the composition of three subqueries  $Q_{init} \circ Q_{trans} \circ Q_{final}$.
$Q_{init}$ states the first configuration is the one starting at the initial state with empty counters : $\epsilon \langle   EQ(\child_{c_1} ,"0") \wedge EQ(\child_{c_2} ,"0") \wedge EQ(\child_{state} ,"q_0") \rangle \circ \child_{next}$
$Q_{final}$ states that the last configuration is in the final state $\epsilon  \langle  EQ(\child_{state} ,"q_f") \rangle$.
$Q_{trans}$ states the transition between two configurations :  $ (\epsilon \langle  \bigvee_{q} \phi_{q} \rangle \circ \child_{next})^+$.
We explain how to encode the different possible transitions over $c_1$, the same can be done on $c_2$:
\begin{itemize}
\item $\delta(q) = (incr_1,q_1)$, then $\phi_q = EQ(\child_{c_1}, \child_{next} \circ \child_{c_1} \circ \child_{a}) \wedge EQ(\child_{state} , "q")  \wedge EQ( \child_{next} \circ \child_{state},"q_1")$
\item $\delta(q) = (del_1,q_2)$, then  $\phi_q = EQ(\child_{c_1}\circ \child_{a}, \child_{next} \circ \child_{c_1} ) \wedge EQ(\child_{state} , "q")  \wedge EQ( \child_{next} \circ \child_{state},"q_1")$
\item $\delta(q) = (if ~0 ~then ~q_1, q_2)$ then 
$$\phi_q = \big(EQ(\child_{c_1},"0" ) \wedge EQ(\child_{state} , "q")  \wedge EQ( \child_{next} \circ \child_{state},"q_1")\big) $$
$$\vee \big( [\child_{c_1} \circ \child_a ] \wedge EQ(\child_{state} , "q")  \wedge EQ( \child_{next} \circ \child_{state},"q_2")\big) $$
$$\bigwedge  EQ(\child_{c_1}, \child_{next} \circ \child_{c_1} )  $$
\end{itemize}

Let assume that the machine $M$ finishes in the final state.
Let $r$ be an accepting run. We construct the associated JSON tree as explained before. Each configuration is represented by an object having four keys : state returning the state of the machine, $c_1$ returning an object coding the value of the counter $C_1$, 
$c_2$ returning an object coding the value of the counter $C_2$ and $next$ returning the object representing the next configuration.
Each object coding the value $n$ of counter is JSON tree coding path of length $n$ such that each object has only one key $a$ at the exception of the last value is the integer $0$.
We can notice that we can unfold the query by following the run such that at each step, the transition is checked.
Let assume that the query is satisfiable. We can extract for each object under the key $next$ a configuration as explained before. The transition between each configuration has the query transition imposes that between two objects representing a configuration, the counters have the correct values as full subtrees are copied.  
\end{proof}

\smallskip
\noindent
\textbf{Proposition} \ref{prop-JNL-rec-sat-dec}
The \textsc{Satisfiability} problem is: 
\begin{itemize}
\item \pspace-complete for non-deterministic, non-recursive JNL without the $\EQ(\alpha,\beta)$ operator. 
\item \exptime-complete for non-deterministic, recursive JNL without the $\EQ(\alpha,\beta)$ operator.
\end{itemize}

\begin{proof}
The \pspace lower bound follows from Proposition \ref{prop-jsl-sat}, and the fact that we can go from JSL to JNL in polynomial time (Theorem \ref{thm-jsn_jnl}). 
The \exptime-hardness follows from adapting known reductions from bounded space alternating turing machines to PDL (see e.g. Harel, D., Kozen, D., Tiuryn, J., 2000. Dynamic Logic. MIT Press.). The upper bound follows 

Both the \pspace-upper bound and the \exptime-upper bound follows from Proposition \ref{prop-jsl-sat} and Proposition \ref{prop-sat-recursive-jsl}.  
Even though we cannot go directly through JSL because this would mean an exponential blowup, we can eliminate this blowup by using a new definition each time we 
concatenate formulas. For the \pspace-hardness this implies introducing definitions, but since these are acyclic we can still bound the height of the accepted JSON documents. 
\end{proof}

\smallskip
\noindent
\textbf{Theorem \ref{schema2logic}}

  \begin{itemize}
\item JSON Schema can be expressed in JSL.
\item String-deterministic JSL can be expressed in JSON Schema.
\end{itemize}

\begin{proof}
 \textbf{Schema to JSL}. Let us start with the first item. For each schema $S$ we proceed to construct a formula $\psi_S$, we do this inductively on the syntax of JSON Schema.

Let us start our construction with strings. Consider the following string schema $S$

{\small\begin{alltt}
                                       \{
                                        "type": "string", 
                                        "pattern": "\(e\)"  
                                        \}
\end{alltt}}
where $e$ is a regular expression over $\Sigma^*$. Here we need to check both that $J$ is an string and that it belongs to the language of $e$. 
We do this with the following formula $\psi_S$ in $JSL$
\begin{align*}
  \psi_S = \isstring\land\patternlogic(e) 
\end{align*}
 Note that in the absence of the pattern restriction we just have to remove the last predicate from $\psi_S$.
 
 \medskip
 
 In the case of numbers we provide a similar construction, consider a the following generic schema $S$
 {\small\begin{alltt}
                                       \{
                                        "type": "number", 
                                        "minimum": "\(n\)",
                                        "maximum": "\(m\)",
                                        "multipleOf": "\(k\)"                                          
                                        \}
\end{alltt}}
 
 where $n,m,k$ are decimal numbers. Here the construction is as follows
 \begin{align*}
  \psi_S = \isnumber\land\minim(n)\land\maxim(m)\land\mof(k)
\end{align*}
Again if one of the restrictions is nto present we can just avoid it in the construction of the formula. 

\medskip

We continue with object schemas, consider the following object schema $S$

{\small\begin{alltt}
                                    \{
                                       "type": "object", 
                                       "minProperties": \(i\),
                                       "maxProperties": \(j\),                                       
                                       "required": ["\(s\sb{1}\)",...,	"\(s\sb{n}\)"],
                                       "properties": \{
                                                 "\(k\sb{1}\)": \(S\sb{1}\),
                                                 ...
                                                 "\(k\sb{m}\)": \(S\sb{m}\)  
                                                 \},
                                       "patternProperties": \{
                                                 "\(r\sb{1}\)": \(D\sb{1}\),
                                                 ...
                                                 "\(r\sb{l}\)": \(D\sb{l}\)  
                                                 \},
                                       "additionalProperties": \(A\)   
                                       \}
\end{alltt}}

Here we have to use our inductive step to assign a formula for each subschema present in the definition of $S$. Let $\phi_{S_1},\ldots,\phi_{S_m},\phi_{D_1},\ldots,\phi_{D_l},\phi_A$ be those formulas, and let $C$ be the intersection of the complement of each expression $k_1,\ldots,k_m, r_1,\ldots, r_l$, we can construct $\psi_S$ as follows
 \begin{align*}
  \psi_S = \isobject\land\minch(i)\land\maxch(j)\land&\\
  \dm_{s_1}\top\land\ldots\land\dm_{s_n}&\top\land\\
  \bx_{k_1}\phi_{S_1}\land&\ldots\land\bx_{k_m}\phi_{S_m}\land\\
    \bx_{r_1}&\phi_{D_1}\land\ldots\land\bx_{r_l}\phi_{D_l}\land\\
    &\bx_{C}\varphi_A
\end{align*}

For arrays consider the following schema $S$
{\small\begin{alltt}
                                       \{
                                       "type": "array",                                      
                                       "items": ["\(S\sb{1}\)",...,"\(S\sb{n}\)"],
                                       "uniqueItems": true,
                                       "additionalItems": \(A\)   
                                       \}
\end{alltt}}

In this case we need to force the JSON values to satisfy their corresponding indexes schema, we achieve this by constructing the following $JSL$ formula
 \begin{align*}
  \psi_S = \isarray\land\uniq\land\dm_{1,1}\phi_{S_1}\land&\ldots\land\dm_{n,n}\phi_{S_n}\land\bx_{j+1:+\infty}\varphi_A
\end{align*}
Where $\phi_{S_1},\ldots, \phi_{S_n}, \phi_{S_A}$ come from the inductive step. Note that when the additionalItems keyword is not present we need to 
introduce instead $\bx_{j+1:+\infty}\bot$, to specify that there cannot be more children.

Finally, we conclude by providing the construction for the boolean operators of JSON Schema, consider the following schemas
{\small
\begin{alltt}
                       \(S\sb{allOf}=\) \{"allOf": [\(S\sb{1}\),...,\(S\sb{n}\)]\}     \(S\sb{anyOf}=\) \{"anyOf": [\(S\sb{1}\),...,\(S\sb{n}\)]\} 

                       \(S\sb{not}=\) \{"not": \(S'\)\}                \(S\sb{enum}=\) \{"enum": [\(J\sb{1}\),...,\(J\sb{n}\)]\}
 \end{alltt}}
The construction comes naturally from the boolean semantics

\begin{align*}
\varphi_{S_{allOf}}&:=\phi_{S_1}\land\ldots\land\phi_{S_n}\\
\varphi_{S_{anyOf}}&:=\phi_{S_1}\lor\ldots\lor\phi_{S_n}\\
\varphi_{S_{not}}&:=\neg\phi_{S_A}\\
\varphi_{S_{enum}}&:=\compare(J_1)\lor\ldots\lor\compare(J_n)\\
\end{align*}
 Where $\phi_{S_1},\ldots, \phi_{S_n}, \phi_{S_A}$ come from the inductive step.
 
 \bigskip
 
 \textbf{JSL to Schema}. For the second item we show how to obtain a schema $S^\phi$ from each JSL formula $\phi$. 
 We proceed again by induction on the structure of JSL formulas. 
 
 \smallskip
 \noindent
 The formula $\top$ is the schema \verb+{}+, that by definition validates against any document. 

   \smallskip
 \noindent
 All $\nodetests$ are similar: 
 \begin{itemize}
 \item If $\phi$ is $\Obj$ then $S$ is just \verb+{"type": "object"}+, and similarly for $\Arr$, $\Int$ and $\Str$
 \item If $\phi$ is $\uniq$ then $S$ is \verb+{"type": "array", "uniqueItems": true}+. The other predicates in $\nodetests$ are completely symmetrical, 
 except for $\minch(i)$ and $\maxch(i)$ that we specify below. Note that we need to specify the type of documents they act on. 
 \item In the case of $\minch(i)$ and $\maxch(i)$ we need to specify a disjunction between arrays and objects. For example, 
 $\minch(i)$ is the union of \verb+{"type": "object", "minProperties": i}+ and \verb+{"type": "array", "items": [{},...,{}]}+, repeating \verb+{}+ $i$ times, and 
 where \verb+{}+ by definition validates against any document. For $\maxch(i)$ is a bit more complicate, as apart from 
 \verb+{"type": "object", "maxProperties": i}+  we need to take the union of each 
 \verb+{"type": "array", "items": [{}], "additionalItems": J }+, where \verb+J+ is an unsatisfiable schema, for each repetition of items between $1$ and $i$. 
 Each of this specifies an array with exactly $k$ elements, where $k$ is the amount of times the schema \verb+{}+ is repeated in the items. 
 \end{itemize}
 
 \smallskip
 \noindent
 Disjunction, conjunction and negation are immediate from the boolean operators in JSON Schema. 
 
 \smallskip
 \noindent
  If $\phi$ is $\dm_{i:j} \psi$ then the schema is: 
 
 \begin{verbatim}
 {
    "type": "array", 
    "items":[{},...,{},Spsi,...,Spsi], 
    "additionalItems: {}
 }
 \end{verbatim}  
 
 where \verb+Spsi+ is the schema for $\phi$, the repetition of \verb+{}+ is given $i-1$ times, and the repetition of \verb+Spsi+ is given $j$ times. 
 
  \smallskip
 \noindent
  If $\phi$ is $\dm_{i:+\infty} \psi$ then the schema is: 
 
  \begin{verbatim}
 {
    "type": "array", 
    "items":[{},...,{}], 
    "additionalItems: Spsi
 }
 \end{verbatim}  
 
  where \verb+Spsi+ is the schema for $\phi$, and the repetition of \verb+{}+ is given $i-1$ times.

  \smallskip
 \noindent
 If $\phi$ is $\bx_{e} \psi$ then the schema is: 
 
 \begin{verbatim}
 {
    "type": "object", 
    "patternProperties:{
       "e": Spsi
       }
 }
 \end{verbatim} 
The rest of the modalities can be simulated using the equivalences  
$\bx_e \psi \equiv \not \dm_e \not \psi$ and $\dm_{i:j} \psi \equiv \not \bx_{i_j} \not \psi$

\end{proof}

\smallskip
\noindent
\textbf{Theorem \ref{thm-jsn_jnl}}:\\ 
Non-deterministic JNL not using the equality $\EQ(\alpha,\beta)$ and non-deterministic JSL using only the node test $\sim(A)$ are equivalent in terms of  expressive power. More precisely:
\begin{itemize}
\item For every formula $\varphi^S$ in JSL exists a unary formula $\varphi^N$ in JNL such that for every JSON $J$: $$\sem{\varphi^N}_J = \{n\in J \mid (J,n)\models \varphi^S\}.$$
\item For every unary  formula $\varphi^N$ JNL exists a $\varphi^S$ in JSL such that for every JSON $J$: $$\sem{\varphi^N}_J = \{n\in J \mid (J,n)\models \varphi^S\}.$$
\end{itemize}
Moreover, going from JSL to JNL takes only polynomial time, while the transition in the other direction can be exponential.

\begin{proof}
The proof of the first item is done by an easy induction on the structure of $\varphi^N$.
\noindent{\underline{Base case:}}
\begin{itemize}
\item If $\varphi^N = \top$ then we define $\varphi^S = \top$.
\item If $\varphi^N = \sim(A)$ then we define $\varphi^S = \EQ(\varepsilon,A)$.
\end{itemize}
\noindent{\underline{Inductive case:}}
\begin{itemize}
\item The boolean cases are handled in an obvious manner. I.e. If $\varphi^N = \neg \varphi_1$ then we define $\varphi^S = \neq \varphi_1^S$.
\item If $\varphi^N = \dm_{e}\ \phi$, then $\varphi^S = [X_e\langle\phi^S\rangle]$.
\item If $\varphi^N = \dm_{i:j}\ \phi$, then $\varphi^S = [X_{i:j}\langle\phi^S\rangle]$.
\item The case of diamond is straightforward, since $\bx_{e}\ \phi$ is equivalent to $\neg \dm_{e}\ \neg\phi$ and similarly for $\bx_{i:j}$.
\end{itemize}
It is now easy to show that $\varphi^S$ has the same semantics as $\varphi^N$.
\smallskip

For the second item a bit more care is needed, since JNL also allows binary formulas, so we have to show how one can keep track of those using JSL formulas which are by definition unary. For this we will need some more notation. 

First, note that when considering a JNL formula $X_a$, its clear analogue in JSL would be $\dm_{a} \top$, since it gives us all the nodes $n$, such that there is $n'$ with $(n,n')\in \sem{X_e}_J$, for any JSON $J$. But what do we do with a formula $X_a\circ X_b$? Here the first part of the composition ($X_a$) is equivalent to $\dm_{a} \top$, and the second to $\dm_{b} \top$. However, we can not simply glue them together, but need to replace the $\top$ in $\dm_{a} \top$, with the entire formula $\dm_{b} \top$, thus obtaining $\dm_{a}\dm_{b} \top$. In order to keep track of which $\top$ we are using, we will assign a different top symbol $\top_\varphi$ to each formula $\varphi$ (marked by the formula itself). Furthermore, when we need to replace a top symbol in the formula $\varphi$ with another formula $\psi$, we will write $\varphi[\top_\varphi \rightarrow \psi]$ to denote this.

Next, we will need to use a special $\top$ symbol for formulas of the form $\langle \varphi \rangle$. At first it might seem that a formula like $\langle \varphi\rangle$ could simply be replaced with the translation of $\varphi$, which we denote by $\varphi^S$, since we only need to check if our node satisfies $\varphi$. However, this might cause problems when we continue composing $\langle \varphi \rangle$ with other binary formulas. For instance, in $\langle X_a \rangle \circ X_b$ this would result with an incorrect translation to $\dm_{a}\dm_{b} \top$. That is, we need to treat formulas of the form $\langle \varphi \rangle$ as existential tests that can be continued afterwards. For this, we will translate formulas of the form $\psi = \langle \varphi \rangle$ to $\psi^S = \top_\phi \wedge \varphi^S[\top_\varphi \rightarrow \top^*]$. Here $\top^*$ is a special top symbol which helps us continue composing correctly using $\psi^S$. That is, when another formula follows $\psi$, we will now know that $\top^*$ will not be replaced by it, thus making it a true existential test as intended.

The full inductive translation which works simultaneously on unary and binary formulas of JNL is given next. We start with the base cases. Here JNL formulas will be denote $\varphi, \psi, \alpha$ and their (intermediate) translations as $\varphi^{S_I}, \psi^{S_I}, \alpha^{S_I}$.

\noindent{\underline{Base case:}}
\begin{itemize}
\item If $\varphi = \top$ then we define $\varphi^{S_I} = \top_\varphi$.
\item If $\alpha = \varepsilon$ then we define $\alpha^{S_I} = \top_\alpha$.
\item If $\alpha = X_e$ then $\alpha^{S_I} = \dm_{e} \top_\alpha$.
\item If $\alpha = X_{i:j}$ then $\alpha^{S_I} = \dm_{i:j} \top_\alpha$.
\end{itemize}
\noindent{\underline{Inductive case:}}
\begin{itemize}
\item If $\varphi = \neg \psi$ then $\varphi^{S_I} = \neg \psi^{S_I}[\top_\psi \rightarrow \top_\varphi]$.
\item If $\varphi = \psi_1 \wedge \psi_2$ then $\varphi^{S_I} = \psi_1^{S_I}[\top_{\psi_1} \rightarrow \top_\varphi] \wedge \psi_2^{S_I}[\top_{\psi_2} \rightarrow \top_\varphi]$.
\item If $\varphi = \psi_1 \vee \psi_2$ then $\varphi^{S_I} = \psi_1^{S_I}[\top_{\psi_1} \rightarrow \top_\varphi] \vee \psi_2^{S_I}[\top_{\psi_2} \rightarrow \top_\varphi]$.
\item If $\varphi = [\alpha]$ then $\varphi^{S_I} = \alpha^{S_I}[\top_\alpha \rightarrow \top_\varphi]$.
\item If $\varphi = \EQ(\alpha,A)$ then $\varphi^{S_I} = \alpha^{S_I}[\top_\alpha \rightarrow \sim(A)]$.
\item If $\alpha = \alpha_1 \circ \alpha_2$ then $\alpha^{S_I} = (\alpha_1^{S_I}[\top_{\alpha_1} \rightarrow \alpha_2^{S_I}])[\top_{\alpha_2}\rightarrow \top_\alpha]$.
\item If $\alpha = \langle \varphi \rangle$ then $\alpha^{S_I} = \top_\alpha \wedge \varphi^{S_I}[\top_\varphi \rightarrow \top^*]$.
\end{itemize}

Finally, for $\varphi$ we define $\varphi^S = \varphi^{S_I}[\{\top^*,\top_\varphi\}\rightarrow \top]$, replacing all the remaining top symbols with true. We do the same thing with binary formulas $\alpha$ giving us $\alpha^S$.

To demonstrate how the translation works, consider the formula $\EQ(\langle X_b \rangle \circ X_a, A)$, which checks that the current node contains a key $b$, and that it contains a key $a$ with the value $A$. Our translation will start by producing $\top_1 \wedge \dm_{b}\top^*$ in place of $\langle X_b \rangle$. It will then replace $\top_1$ with $\dm_{a}\top_2$, i.e. with the translation of $X_a$. Finally, it will replace $\top_2$ with $\sim(A)$, producing $\dm_{a}\sim(A)\wedge \dm_{b} \top$, as desired.

Note that in the formulas of the form $\langle [X_{a_1}] \vee [X_{a_2}] \rangle \circ \langle [X_{b_1}] \vee [X_{b_2}] \rangle \circ \langle [X_{c_1}] \vee [X_{c_2}] \rangle$ our translation keeps track of all the possible paths specified by the JNL formula, thus resulting in an exponentially long JSL formula.

Using a rather cumbersome induction (which works simultaneously on unary and binary JNL formulas), we can now prove the claim that $\sem{\varphi}_J = \{n\in J \mid (J,n)\models \varphi^S\}$, for any JNL formula $\varphi$. Note that here we also need to prove an auxiliary claim on binary formulas, which states that $(J,n)\models \alpha^S$ iff $\exists n'$ s.t. $(n,n')\in \sem{\alpha}_J$, and $n'$ is used as a witness to the truth of $\alpha^S$ in $n$.
\end{proof}

\smallskip
\noindent
\textbf{Proposition} \ref{prop-jsl-eval}: \\
The \textsc{Evaluation} problem for JSL can be solved in time $O(|J|^2 \cdot |\phi|)$, and in 
$O(|J| \cdot |\phi|)$ when $\phi$ does not use the uniqueItems keyword. 

\begin{proof}
We can use Theorem \ref{thm-jsn_jnl} to transform the JSL expression into a JNL expression without the equality operator, but with one additional unary predicate for each 
node test in $\nodetests$. All these predicates can be pre-computed in time $O(|J|\cdot|\phi|)$ except for $\uniq$, which needs $O(|J|^2)$. 
\end{proof}

\smallskip
\noindent
\textbf{Proposition \ref{prop-jsl-sat}}:\\
The \textsc{Satisfiability} problem for JSL is in \expspace, and \pspace-complete for expressions without $\uniq$. 

\begin{proof}
The \pspace-hardness is established by reduction from the satisfiability problem of a quantified boolean formula (QBF) in 3CNF, which is 
widely known to be \pspace-complete. More precisely, a QBF in 3CNF is an expression 
of the form $Q_1\ x_1\ \cdots Q_n\ x_n \phi$, such that $\phi$ is a propositional 
formula over $\{x_1,\dots,x_n\}$ in conjunctive normal form and where each clause has 3 literals, 
and each $Q_i$ is either $\forall$ or $\exists$. The satisfiability problem asks if the resulting 
quantified formula is satisfiable. 

The reduction is inspired by a similar reduction in (Benedikt, M., Fan, W., \& Geerts, F. (2008). XPath satisfiability in the presence of DTDs. Journal of the ACM (JACM), 55(2), 8.).
The main idea is to construct a formula such that each of its models correspond to one of the 
possible assignments for the existentially quantified variables of $\phi$, and that contains all possible assignments 
for the universally quantified variables. 

Let then $\alpha$ be a QBF expression in 3CNF of the form $Q_1\ x_1\ \dots Q_n\ x_n \phi$ as explained above. We construct a JSL formula $\phi$ such that 
$\phi$ is satisfiable if and only if $Q_1\ x_1\ \dots Q_n\ x_n \phi$ is satisfiable. 

Our formula is defined as $\phi = \phi_\text{tree} \wedge \phi_\text{clauses}$, where each of these subformulas are defined as follows. 

Formula $\phi_\text{tree}$ is a formula that mandates all the models to be trees of length $2n$, 
consisting only of object nodes, where all nodes at height $0 \leq 2k < 2n$ are objects with a single child through an edge labelled with $X$, 
and each node at height $2k+1$ is an object that has either a single child with an edge labelled with any of $T$ or $F$, if 
$Q_i$ is the existential quantifier $\exists$ or exactly two children, one with an edge labelled with $T$ and the other 
with an edge labelled with $F$. For example, a possible model of the formula 
$\exists x_1 \forall x_2 \forall x_3 (x_1 \wedge x_2 \wedge x_3)$ is the JSON shown below (all nodes are objects). 

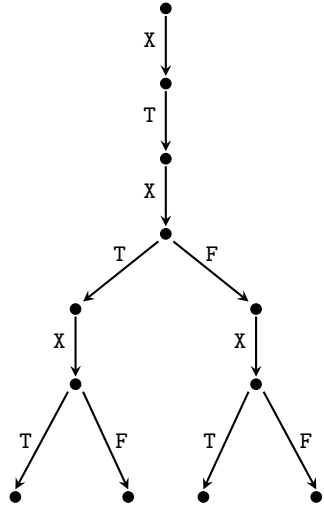
\begin{center}
\begin{tikzpicture}[>=stealth]
\filldraw [black] (0,0) circle (2pt)
(0,-1) circle (2pt)
(0,-2) circle (2pt)
(0,-3) circle (2pt)
(-1.2,-4) circle (2pt)
(1.2,-4) circle (2pt)
(-1.2,-5) circle (2pt)
(1.2,-5) circle (2pt)
(-2.0,-6.5) circle (2pt)
(-0.5,-6.5) circle (2pt)
(0.5,-6.5) circle (2pt)
(2,-6.5) circle (2pt)

;

\path (0,-0.1) edge[->, thick] node[pos=0.4,left=0] {\small \texttt{X}} (0,-0.9);
\path (0,-1.1) edge[->, thick] node[pos=0.4,left] {\small \texttt{T}} (0,-1.9);
\path (0,-2.1) edge[->, thick] node[pos=0.4,left] {\small \texttt{X}} (0,-2.9);

\path (-0.1,-3.1) edge[->, thick] node[pos=0.2,left=0.1] {\small \texttt{T}} (-1.1,-3.9);
\path (0.1,-3.1) edge[->, thick] node[pos=0.2,right=0.1] {\small \texttt{F}} (1.1,-3.9);

\path (-1.2,-4.1) edge[->, thick] node[pos=0.4,left] {\small \texttt{X}} (-1.2,-4.9);
\path (1.2,-4.1) edge[->, thick] node[pos=0.4,left] {\small \texttt{X}} (1.2,-4.9);

\path (-1.3,-5.1) edge[->, thick] node[pos=0.5,left] {\small \texttt{T}} (-2.0,-6.4);
\path (-1.1,-5.1) edge[->, thick] node[pos=0.5,right] {\small \texttt{F}} (-0.5,-6.4);

\path (1.1,-5.1) edge[->, thick] node[pos=0.5,left] {\small \texttt{T}} (0.5,-6.4);
\path (1.3,-5.1) edge[->, thick] node[pos=0.5,right] {\small \texttt{F}} (2.0,-6.4);


\end{tikzpicture}
\end{center}

In order to make sure each tree has this form, we let $\phi_\text{tree}$ be the conjunction of all the formulas we describe next.  

\begin{itemize}
\item For each $0 \leq k \leq n$ such that $k$ is existentially quantified, 
the formula $$(\bx_{\Sigma^*})^{2(k-1)} \bx_{X} \big((\dm_T \top \wedge \neg \dm_F \top) \vee (\neg \dm_T \top \wedge \dm_F \top),$$
where $(\psi)^{i}$ represents the concatenation of $\psi$ $i$ times. 
\item For each $0 \leq k < n$ such that $k$ is universally quantified, 
the formula $$(\bx_{\Sigma^*})^{2(k-1)} \bx_{X} \big((\dm_T \top \wedge \dm_F \top).$$
\item For each $0 \leq k < n-1$, formula 
$$(\bx_{\Sigma^*})^{2k+1} \big(\bx_T\dm_X\top  \wedge \bx_F\dm_X\top)$$
\end{itemize}

As we have mentioned, for each assignment of the existentially quantified variables in $\alpha$ there is a 
JSON tree satisfying all these formulas and where the path from the root to each leaf 
represent a possible assignment of all variables in $\alpha$. 

Formula $\phi_\text{clause}$ makes sure that all such assignments satisfy $\phi$. To do this consider, for 
each clause $C = (\star x_i \vee \star x_j \vee \star x_k)$, where $\star$ is either $\neg$ or nothing, the formula $\phi_C$ defined as 

$$\phi_C = (\dm_{\Sigma^*})^{2(i-1)} \dm_X\dm_{V_i} (\dm_{\Sigma^*})^{2(j-i-1)} \dm_X\dm_{V_j} (\dm_{\Sigma^*})^{2(k-j-i-1)} \dm_X\dm_{V_k}\top, $$
where $V_\ell$ is $T$ if $\star x_\ell$ is $\neg x_\ell$ or $F$ otherwise. 

Note that this formula is true only when there is a path from the root to a leaf such that there is a $V_i$-labelled edge after the node at height $2i+1$, 
a  $V_j$-labelled edge after the node at height $2j+1$ and a $V_k$-labelled edge after the node at height $2k+1$. 
Thus, this holds when one can find a path in which the assignment represented by this path does not satisfies clause $C$. 

Define then $\phi_\text{clause} = \bigwedge_{C \in \alpha} \neg \phi_C$.  This formula is satisfied when, for all clauses $C$ of $\phi$, it is impossible to find 
a path in which the assignment represented by this path does not satisfies clause $C$. 
 It is now not difficult to show that $\alpha$ is satisfiable if and only if the expression 
$\phi_\text{tree} \wedge \phi_\text{clause}$ is satisfiable.    

The upper bound follows from Proposition \ref{prop-sat-recursive-jsl} and the remark that non-recursive JSL formulas only accepts trees 
of height bounded polynomially by the size of the formula. With this remark, we can use a technique introduced in 
(Benedikt, M., Fan, W., \& Geerts, F. (2008). XPath satisfiability in the presence of DTDs. Journal of the ACM (JACM), 55(2), 8.) 
to deal with the satisfiability of non-recursive DTDs in the XML context. In a nutshell, the idea is to simulate a tree automata with a 
string automata, using a coding of trees of bounded height into strings. If $\uniq$ is not present then we can preprocess the formula first to know 
the maximum width of the trees we will need to check, but if $\uniq$ is indeed present we end up with an automaton of exponentially many rules 
since for the encoding we need to check an exponential amount of children for each node. 
\end{proof}

\smallskip
\noindent
\textbf{Theorem \ref{schema2logic-recursive}}:\\
 Well-formed recursive JSL and well-formed recursive JSON Schema are equivalent in expressive power. 
\begin{proof}
Follows from the immediate definition of recursive JSL and Theorem \ref{schema2logic}. From JSON Schema to 
JSL all we need to do is to encompass each of the definitions as a new definition in the recursive JSL formula, and 
translate the base JSON Schema as the base expression. The other direction is symmetrical. 
\end{proof}


\smallskip
\noindent
\textbf{Proposition \ref{prop-jsl-mso}}:\\
Well-formed recursive JSL is equivalent to MSO, if the set $\Sigma^*$ of possible keys is 
fixed and finite, and that we do not consider arrays.
\begin{proof}
When no arrays exists we can express all of the remaining $\nodetests$ with MSO (the only problem is $\uniq$), and we can express all the modalities as well. It is thus 
immediate to show that MSO can encode JSL. Furthermore, to encode a well-formed recursive JSL expression $\Delta$, for each definition $\gamma = \phi$ in $\Delta$ 
we create a set $S_\gamma$ and the formula $S_\gamma(x) \leftrightarrow T(\phi)$, where $T(\phi)$ is the translation of $\phi$, assuming each symbol $\gamma'$ from $\Gamma$ in 
$\phi$ is translated as $S_{\gamma'}$. 

To go from MSO to JSL, the key thing to note is that now each object has at most $k$ children, where $k$ is the fixed number of keys. Thus MSO can be captured by an automata of 
over trees of rank $k$ (see e.g. Libkin, Leonid. Elements of finite model theory. Springer Science \& Business Media, 2013.). In turn, we can use JSL to capture these types of automata. See \cite{PRSUV16} for an explanation on how to encode automata with JSON Schema; 
the translation therein can be seemingly adapted for trees of rank $k$. 
\end{proof}

\smallskip
\noindent
\textbf{Proposition} \ref{prop-eval-recursive-jsl}:\\
The \textsc{Evaluation} problem is \ptime-complete for recursive JSL expressions. 

\begin{proof}

For the \ptime\ upper bound we need some notation. 
Let $J$ be a JSON document, and assume that the height of $J$ is $h$. To each symbol $\gamma$ in $\Gamma$ and $k \leq h$ we associate a 
set $S_k^J(\gamma)$ of nodes from $J$, whose intuition is to specify the nodes of $J$ that satisfy $\gamma$, and that are of height at least $k$. 
More formally, let $\mathcal S_h$ be the set of predicates containing the predicate name $S_k(\gamma)$ for each $0 \leq k \leq h$ and each $\gamma \in \Gamma$. Then the satisfaction relation $\models$ can be extended for 
JSL expressions over the extended syntax that allows symbols from $\mathcal S_h$ as predicates, defining that $(J,n) \models S_k(\gamma)$ 
if and only if $n \in S_k^J(\gamma)$. 

Now let $\Delta$ be a recursive JSL expression of the form (\ref{eqn-rec-jsl}). For each $1 \leq j \leq n$, let $\phi_j^h$ be the formula constructed by replacing each symbol $\gamma_i$ in $\phi_j$ under the scope of a modal 
operator by $\bot$, and each other $\gamma_i$ by $S_k(\gamma)$. Further, for each $0 \leq k < h$, let $\phi_j^k$ be the formula constructed by replacing each symbol $\gamma_i$ in $\phi_j$ under a modal operator for $\bigvee_{k < \ell <h} S_\ell(\gamma)$, and each other symbol $\gamma_i$ for $S_k(\gamma_i)$. 

We define for each $0 \leq k \leq h$ the operators given by: 
$$\begin{array}{l}
S_k^J(\gamma_1) = \{n \mid \text{the height of $n$ is $k$, and } \models \phi_1^k\}, \\
\ \ \ \ \vdots \\
S_k^J(\gamma_n) =  \{n \mid \text{the height of $n$ is $k$, and } \models \phi_n^k\}
\end{array}$$

Assume for readability that the precedence graph of $\Delta$ induces the order $\gamma_1 < \gamma_2 < \dots <  \gamma_n$. 
We compute all the sets $S_k^J(\gamma_i)$, following the operators introduced above, from the biggest $\gamma$ symbol to the lowest 
and from height $h$ to height $0$. That is, we start in  
$S_h^J(\gamma_n)$, $S_h^J(\gamma_{n-1}), \dots, S_h^J(\gamma_1)$, and then $S_{h-1}^J(\gamma_n)$, and so on until 
$S_0^J(\gamma_1)$. 
Finally, to compute the nodes of $J$ that satisfy $\Delta$ we let $\psi^h$ be the formula that replaces each symbol $\gamma$ 
in $\psi$ for $\bigvee_{0 \leq k \leq h} S_k(\gamma)$. A routine induction on height shows the following Lemma, stating that 
this semantics does indeed coincide with the semantics based on rewriting.


\begin{lemma}
Let $J$ be a JSON tree, $n$ a node of $J$ and $\Delta$ a recursive JSL expression. Then $(J,n) \models \Delta$ if and only if 
$(\json(n),n) \models \psi^h$, where $h$ is the height of $\json(n)$. 
\end{lemma}

We can now show the \ptime\ upper bound. In order to compute all nodes that satisfy each symbol $\gamma$ we need 
to evaluate one of the definitions of $\Delta$ at most once per each node of $J$. This is in polynomial time (in fact, quadratic in the size of $J$) 
due to Proposition \ref{prop-jsl-eval}.  Once we have finished this operation we can proceed to evaluate the base formula of $\Delta$, which requires another 
polynomial pass. 

The \ptime\ lower bound is from reduction from the problem of validating whether an assignment satisfies a boolean circuit. 
Let $C$ be a boolean circuit with input gates $IN_1,\dots,IN_n$, and gates $G_1,\dots,G_m$, with $G_m$ being the output gate. 
We construct a JSON document that is an object with $n$ key-value pairs, with keys $\texttt{"IN1"}, \texttt{"IN2"},\dots,\texttt{"INN"}$ each of which 
associated to the value \texttt{"T"} is the input is given a $1$ or \texttt{"F"} otherwise. 

For each gate $G_j$ we create a symbol $\gamma_j$ and a definition $\gamma_j = \phi_j$, with $\phi_j$ the codification of the circuit rooted at 
gate $G_j$. If this circuit uses an input gate $IN_i$ we use instead the expression $\dm_{INi} \patternlogic(T)$ that is true if the 
key-value pair of $J$ whose key is \texttt{"INi"} corresponds to the string "T", or false otherwise. 

Finally, the base expression of $\Delta$ is simply $\gamma_m$, corresponding to the output gate. We have that $\Delta$ accepts $J$ if and only if 
$C$ is true under the assignment of the input gates. 
\end{proof}

\medskip
\noindent
\textbf{Proposition \ref{prop-sat-recursive-jsl}}:\\
The \textsc{Satisfiability} problem for recursive JSL expressions is in \twoexptime, and \exptime-complete 
for expressions without the $\uniq$ predicate. 

\begin{proof}
\newcommand{\booleanstuff}{\textit{BoolSNT}}
\newcommand{\boolquant}{\textit{BoolS}}
\newcommand{\states}{Q_n \cup Q_t}
\newcommand{\quantstates}{\textit{QuantStates}}

The \exptime-lower bound follows from the fact that JSON Schema can simulate tree automata (see \cite{PRSUV16}) and 
that we can capture JSON Schema with JSL. 

\bigskip

For the upper bound we construct an alternating tree automata model that captures recursive JSL and works directly over JSON trees. 

A $J$-automata is a tuple $(Q_n,Q_t,F,\delta)$, where $Q_n$ and $Q_t$ are existential sets of node states and tree states, $F\subseteq \states$ is the 
set of final states and $\delta$ is the transition relation. To explain how $\delta$ is specified we need some notation.  
Let $\quantstates$ be the set of predicates in the sets $\{q \exists_e \mid q \in \states\}$, $\{q \forall_e \mid q \in \states\}$, for 
each regular expression $e$ over $\Sigma$, and $\{q \exists_{i:j}  \mid q \in \states\}$, $\{q \forall_{i:j}  \mid q \in \states\}$ for 
$i \leq j$ or $j = + \infty$.

Let $\boolquant$ be the set of all positive boolean combinations of predicates in $\quantstates$ and let $\booleanstuff$ be the set of all positive boolean combinations of 
predicates in $\states$, $\nodetests$ and $\nodetests^\neg = \{\neg P \mid P \in \nodetests\}$. 

The transition relation $\delta$ is given as rules: 
\begin{itemize}
\item For each $q \in Q_n$, $\delta$ contains a rule of the form $\theta \rightarrow q$, for $\theta \in \booleanstuff$. 
\item For each $q \in Q_t$, $\delta$ contains a rule of the form $\eta \rightarrow q$, for $\eta \in \boolquant$. 
\end{itemize}

We also impose the restriction that there are no loops in state rules: If we form a graph 
whose nodes are $Q_n$ and there is an edge from $p$ to $q$ if there is a rule $\theta \rightarrow q$ where $\theta$ uses 
$p$, then this graph must be acyclic. 

\smallskip
\noindent
\textbf{Semantics}
Given a JSON $J$ and a $J$-automata $A$, a run for $A$ over $J$ augments each node $n$ of $J$ with 
a set $s(n)$ of states and a list $\ell(n) = s_0,\dots,s_k$ of sets of states, such that $s_k = s(n)$, 
the first element of $\ell(n)$ contains only states in $Q_t$, $s_i \subsetneq s_{i+1}$ and 
$s_k \setminus s_0$ contains only states in $Q_n$. 

Let $\rho$ be a run of $A$ over $J$. 
Before referring to accepting runs we need some terminology. 
Let $\theta$ be a formula in $\booleanstuff$. We say that a node $n$ and a set $s$ of states satisfies $\theta$, and write $(n,s) \models \theta$ 
if $(J,n) \models \theta$ when we assume that $(J,n) \models q$ for each state $q \in s$ and $(J,n) \not\models q$ for each state $q \notin s$. 

Let now $\eta$ be a formula in $\boolquant$. We say that a node $n$ satisfies $\eta$, and write $(n,J,\rho) \models \eta$,  
if $(J,n) \models \eta$ when we assume that (1) $(J,n) \models q \forall_e$ if for each triple $(n,w,n')$ in $\Ochild$ with $w \in L(e)$ 
we have that $q \in s(n')$, (2) $(J,n) \models q \forall_{i:j}$ if for each triple $(n,h,n')$ in $\Achild$ with $i \leq h \leq j$ 
we have that $q \in s(n')$, (3)  $(J,n) \models q \exists_e$ if there is a triple $(n,w,n')$ in $\Ochild$ with $w \in L(e)$ 
such that $q \in s(n')$, and (4) $(J,n) \models q \exists_{i:j}$ if there is a triple $(n,h,n')$ in $\Achild$ with $i \leq h \leq j$ 
such that $q \in s(n')$.

We say that $\rho$ is valid and accepting when the following conditions hold: 
\begin{itemize}
\item $s(0)$ contains a final state (assume that $0$ is the root of $J$). 
\item For each node $n$, asume that $\ell(n) = s_0,\dots,s_k$. Then 
\begin{enumerate}
\item A state $q$ is in $s_0$ if and only if there is a rule $\eta \rightarrow q$ such that 
$(n,J,\rho) \models \eta$. 
\item For each $i <0$, a state $q$ is in $s_i$ if and only if  there is a rule $\theta \rightarrow q$ such that 
$(n,J,\rho) \models \theta$.  
\end{enumerate} 
\end{itemize}

An automata $A$ accepts a JSON $J$ if there is a valid and accepting run for $A$ over $J$.

Note that our automata, as any other alternating automata, can easily be complemented in polynomial time. 
To complement a $J$-automata $A$ we just need to swap final with non-final states, and \emph{negate} the formulas. 
For formulas in $\booleanstuff$ this implies swapping conjunctions for negations and negating or un-negating 
predicates in $\nodetests$. For formulas in $\boolquant$ we need to swap existential for universals and conjunctions for disjunctions. 

The following two lemmas show that for our proof it suffices to show that nonemptiness of $J$-automata is in our desired 
upper bounds: 

\begin{lemma}
For every non-recursive JSL formula $\phi$ one can construct in polynomial time a $J$-automata $A^\phi$ such 
that, for every JSON $J$, $J \models \phi$ if and only if $A^\phi$ accepts $J$
\end{lemma}
\begin{proof}
We proceed by induction. 

\smallskip

If $\phi$ is a predicate $P$ in $\nodetests$, then the automata has a single final state $q_f$ and the state transition 
$P \rightarrow q_f$. If $\phi = \top$ then the automata has a single final state and the transition 
$true \rightarrow q_f$, where $\true$ is a special predicate that we assume always true. 

\smallskip 

If $\phi = \psi_1 \wedge \psi_2$, then the automata $A^\phi$ is defined as the union of $A^{\psi_1}$ and $A^{\psi_2}$ with an additional 
final state $q_f$ and a new rule $q_f^1 \wedge q_f^2 \rightarrow q_f$, where $q_f^1$ and $q_f^2$ are the final states of $A^{\psi_1}$ and $A^{\psi_2}$, 
respectively. The case when $\phi = \psi_1 \vee \psi_2$ is similar. 
When $\phi = \neg \psi$, the automaton for $A^{\phi}$ is the complement of $A^{\psi}$

\smallskip

When $\phi$ is $\bx_e \psi$, we construct $A^{\phi}$  by taking $A^{\psi}$ and adding a new state $q_f$ as the unique final state, adding 
the transition $q_f' \forall_e \rightarrow q_f$, where $q_f'$ was the final state of $A^\psi$. 
\end{proof}

\begin{lemma}
For every well-formed recursive JSL formula $\Delta$ one can construct in polynomial time a $J$-automata $A^\Delta$ such 
that, for every JSON $J$, $J \models \Delta$ if and only if $A^\Delta$ accepts $J$
\end{lemma}
\begin{proof}
Assume $\Delta$ is of the form 
\begin{eqnarray}
\nonumber
\gamma_1 &= &\phi_1 \\
\nonumber
\gamma_2 &= &\phi_2 \\
\nonumber
&\vdots &\\
\nonumber
\gamma_m &=& \phi_m \\
\psi & &  \label{eqn-rec-jsl-theo}
\end{eqnarray}
Construct $A^{\phi_1},\dots,A^{\phi_m},A^\psi$ as shown in the previous lemma, except that we assume that $\gamma_1,\dots,\gamma_m$ are part of 
$\nodetests$ (assume the states of these automata are disjoinct). Let $q_1,\dots,q_m$ be fresh states, and let $A^\Delta$ be the automaton constructed from the union of each of these automata, except that we replace each occurrence of $\gamma_i$ in a rule of $\hat A$ for the state 
$q_i$, and that the only final state is the final state coming from $A^\psi$. 
From the evaluation strategy in Proposition \ref{prop-eval-recursive-jsl} it is now not difficult to see that this automata does indeed captures $\Delta$.  
\end{proof}

We can now show how to solve non-emptiness in \twoexptime. Let $A$ be a $J$-automaton. 

\noindent
\textbf{Memory}. We maintain, for each subset of $\states$, the information of whether (a): this subset is reachable or not, and 
(b) how many different trees can be used to reach this state. 

\noindent
\textbf{Initialisation with subsets of $Q_n$}. Assume that $Q_n= \{q_1,\dots,q_m\}$, and let 
$\theta_1,\dots,\theta_n$ be the sets associated to the rules of each of these states. Note that these rules must not mention  any state at all (otherwise 
we can not reach them as leafs from the emptyset). To test whether a subset $S$ of $Q_n$ is reachable, 
let $\theta*$ be the formula $\bigwedge_{q_i \in S} \theta_i \wedge \bigwedge_{q_j \notin S} \not \theta_j$. 
Checking whether $\theta^*$ is satisfiable is clearly in \np, and further, given a number $p$, testing 
whether $\theta^*$ admits more than $m$ different trees can be solved in \ptime with an \np\ oracle. Thus, if $p$ is the largest number associated 
to any predicate $\minch(\cdot)$ in any rule in $A$, we initialise $S$, informing the number of different trees satisfying $\theta$, counting up to the maximum of $p$ different trees (even if $p$ is in binary we can do this in \pspace). 

\noindent
\textbf{General reachability testing}. To check whether other bigger subsets of $Q_n$ are reachable, we just do one pass over all of our subsets, 
and for each of them we try to see if a bigger subset can be reached. 
We do have to modify our function $\theta^*$ accordingly, because we need 
to check again for cardinalities. 

We can also try to incorporate a subset of $Q_t$: for each state $q$ in a subset $s$ of $Q_t$, we need to inspect all subsets of $\states$ that 
we already know can be reached, and understand whether we can construct a tree with $1,\dots,p$ different children. 
For example, assume that we are just checking a single 
state set $\{q\}$. Then let $q_1,\dots,q_n$ be all the sets appearing in the formula $\eta$ of $q$. To understand whether we can form a 
tree of $p$ different children we just choose $p$ subsets of $\states$ amongst the one we already know how to reach, multiplying the cardinality of each of them (because we can combine them). Note that we cannot select two subset of $\states$ such that one is a subset or the other, because this may 
fool our cardinality. 

\noindent
\textbf{Termination}. The process finishes when we reach a final state, or we cannot find any other subset to add. Note that we only require a number of 
steps corresponding to the amount of subsets of $\states$, because in each step either we add a new set or we terminate.  

\noindent
\textbf{Running Time}. We maintain an exponential set of states, and when we are deciding when to reach a new set of states we may possibly have to investigate up 
to $p$ new combinations. Since $p$ is in binary, this means an exponential number of combinations, which produces the second blowup. 
 
When $\uniq$ is not present we can save us the trouble of counting how many different trees satisfy the given conditions. We still have to maintain, when we reach 
a new set of states, that we can reach this set of states with a tree within a certain interval of children. However, each interval can be stored in polynomial space, and 
we need to remember at most a polynomial number of intervals: each new node test or modality introduces just one more interval (when positive) or splits in the worst 
case one interval into two (when negative). So the algorithm runs in \exptime\ in this case. 

\end{proof}

\end{document}